\documentclass[a4paper,11pt]{article}
\usepackage{amssymb}
\usepackage{amsfonts}
\usepackage{graphicx}
\usepackage{makeidx}
\usepackage{amsmath}
\usepackage{natbib}
\usepackage[a4paper]{geometry}
\usepackage{lscape}

\newtheorem{algorithm}{Algorithm}

\newtheorem{lemma}{Lemma}

\newenvironment{proof}[1][Proof]{\noindent \textbf{#1.} }{\  \rule{0.5em}{0.5em}}
\makeindex
\input{tcilatex}
\begin{document}

\begin{titlepage}
\title
{On Bayesian Filtering for Markov Regime Switching Models\thanks{This paper should not be reported as 
representing the views of Norges Bank. The views expressed are those of the authors and do not 
necessarily reflect those of Norges Bank.}}
\author{Nigar Hashimzade\thanks
{Department of Economics and Finance, Brunel University,
United Kingdom; e-mail: nigar.hashimzade@brunel.ac.uk}
\and Oleg Kirsanov\thanks
{Adam Smith Business School, University of Glasgow,
United Kingdom; e-mail: oleg.kirsanov@glasgow.ac.uk}
\and Tatiana Kirsanova\thanks
{Adam Smith Business School, University of Glasgow,
United Kingdom; e-mail: tatiana.kirsanova@glasgow.ac.uk}
\and Junior Maih\thanks{Norges Bank; e-mail Junior.Maih@norges-bank.no}}
\date
{\today}
\maketitle
\begin{abstract}
This paper presents a framework for empirical analysis of dynamic macroeconomic models using Bayesian filtering, with a specific focus on the state-space formulation of Dynamic Stochastic General Equilibrium (DSGE) models with multiple regimes. We outline the theoretical foundations of model estimation, provide the details of two families of powerful multiple-regime filters, IMM and GPB, and construct corresponding multiple-regime smoothers. A simulation exercise, based on a prototypical New Keynesian DSGE model, is used to demonstrate the computational robustness of the proposed filters and smoothers and evaluate their accuracy and speed for a selection of filters from each family. We show that the canonical IMM filter is faster and is no less, and often more, accurate than its competitors within IMM and GPB families, the latter including the commonly used Kim and Nelson (1999) filter. Using it with the matching smoother improves the precision in recovering unobserved variables by about 25\%. Furthermore, applying it to the U.S. 1947-2023 macroeconomic time series, we successfully identify significant past policy shifts including those related to the post-Covid-19 period. Our results demonstrate the practical applicability and potential of the proposed routines in macroeconomic analysis.
\end{abstract}
\par
Keywords: Markov switching models, Filtering, Smoothing
\par
JEL Reference Number: C11, C32, C54, E52
\end{titlepage}\baselineskip=18pt

\section{Introduction}

In the evolving landscape of macroeconomic analysis, the empirical
examination of dynamic models has become increasingly sophisticated and
computationally demanding. This paper contributes to this area by presenting
a comprehensive framework for empirical analysis of state space models with
multiple regimes using Bayesian filtering. Our work introduces enhanced
filter and smoother algorithms, crucial for accurate macroeconomic modelling
and estimation.

Our study is motivated by the increasing popularity of Bayesian methods in
macroeconomic time series analysis in the Dynamic Stochastic General
Equilibrium (DSGE) framework, usually presented in the state-space form.
These methods have gained traction due to their ability to effectively
handle complex models with latent variables and structural changes. Bayesian
perspective is invaluable for disentangling convoluted macroeconomic
phenomena, such as differentiating between external shocks and policy-driven
economic patterns.

Despite significant advancements in the literature, the field continues to
face various challenges, especially when estimating macroeconomic dynamic
models with multiple regimes. One such challenge is selecting an efficient
and accurate filter for likelihood computation. Another challenge is the
task of reconstructing latent variables through the smoothing of estimated
state variables and regime probabilities.

The prevalent use of the Kim and Nelson (\citealp{Kim1994}, %
\citealp{kim1999state}) filter in macroeconomic applications (see, \textit{%
inter alia}, \citealp{DavigDoh2014}, \citealp*{ChangMaihTan2021},\newline
\citealp*{CLL2022}) suggests limited exploration of alternative methods in
this field. Despite its unquestionable power, Kim and Nelson filter is known
to have certain flaws. Namely, it is computationally intensive and, when
extended to smoothing algorithms, computationally unstable. Perhaps, the
latter is the reason for scant use of smoothing for more accurate recovery
of latent variables in multiple regime models in the existing economic
literature.

Our paper makes both theoretical and empirical contributions in this domain.
First, we introduce and extend the Interactive Multiple Model (IMM) filter,
originally developed by Bar-Shalom (\citealp{BlomBarshalom1988}). Despite
its recognition in the engineering literature, the IMM filter remains
underutilised in economic applications. In addition, we claim that the Kim
and Nelson filter belongs to the family of Generalised Pseudo-Bayesian (GPB)
filters and present it in a general form that accommodates different orders
of approximation. Finally, we develop a computationally stable and easily
implementable smoothing algorithm that can be conveniently adapted to a wide
range of filters in multiple regime setting. Empirically, we apply these
methods to a prototypical New Keynesian DSGE model and the U.S.
macroeconomic time series spanning from 1947 to 2023. This exercise succeeds
in identifying significant policy shifts, particularly in the post-Covid-19
era, and thus demonstrates the practical relevance of our methods.

We validate the superiority of the proposed filter-smoother algorithm using
rigorous simulation exercises. Our findings indicate that the IMM filter
outperforms the Kim and Nelson filter in terms of computational speed while
maintaining comparable accuracy. Moreover, the implementation of our
proposed smoother significantly enhances the precision in the recovery of
latent variables, with an approximate 25\% reduction in estimation errors.
These empirical insights reveal the importance of smoothing in this
framework, overlooked in the existing literature.

One should note that while the Kim and Nelson filter has dominated the
analysis of multiple-regime macroeconomic models, there have been a few
exceptions. \citet*{LiuWangZha2013}{\ apparently }applied IMM to study the
role of land-price dynamics in macroeconomy. \citet{BinningMaih2015} used
IMM to show how certain non-linear filters can be adapted to the multiple
regime setting. \citet*{BLM2018} applied it to study the interplay between
oil price shocks and macroeconomic instability. More recently, \citet*{%
LKMR2024} used IMM in a study of monetary and fiscal policy changes in the
United States. We are unaware of other IMM applications in macroeconomics to
date.

All computations presented in this paper were implemented in the RISE$^{%
\textcopyright}$ toolbox (\citealp{MaihRISE2015}).\footnote{%
RISE stands for `Rationality in Switching Environments'. The codes and
documentation are available at https://github.com/jmaih/RISE\_toolbox}

The paper is organised as follows. The next section presents theoretical
foundations. We derive two families of filters, one of which encompasses the
Kim and Nelson filter and the other one encompasses the canonical IMM. We
derive a Markov-switching smoother adapted to the appropriate filter family.
Section \ref{Sec Simulations} tests the efficacy of the proposed filter and
smoother algorithms on artificial data. An empirical investigation is
presented in Section \ref{Sec Empirics}. Section \ref{Sec Conclusions}\
concludes.

\section{Switching Filters and Smoothers\label{Sec Theory}}

\subsection{The Filtering Problem}

We start with a general multiple-regime state-space representation of a
linear discrete-time dynamic model consisting of a measurement equation (\ref%
{me}) and a transition equation (\ref{te}), 
\begin{eqnarray}
y_{t} &=&c_{y,s_{t}}+Z_{s_{t}}\alpha _{t}+g_{s_{t}}\varepsilon _{t},
\label{me} \\
\alpha _{t} &=&c_{\alpha ,s_{t}}+T_{s_{t}}\alpha _{t-1}+R_{s_{t}}\eta _{t},
\label{te}
\end{eqnarray}%
where $y_{t}$ a $p\times 1$ vector of observations, $\alpha _{t}$ is a $%
m\times 1$ vector of unobserved state variables, and $\varepsilon _{t}$ and $%
\eta _{t}$ are independent standard Gaussian random variables, $t=1,\ldots
,n $. All model parameters, $\left \{ c_{y,s_{t}},c_{\alpha
,s_{t}};Z_{s_{t}},T_{s_{t}};g_{s_{t}},R_{s_{t}}\right \} $, depend on regime 
$s_{t}$, which is an outcome of a Markov process with $h\geq 1$ discrete
regimes. This process is described by the transition probability matrix with
the generic element $Q(s_{t-1},s_{t})=\func{Pr}\left[ s_{t}\mid s_{t-1}%
\right] $, so that $\sum_{s_{t}=1}^{h}Q(s_{t-1},s_{t})=1$ for every regime $%
s_{t-1}$ and every time $t$.

The information available at time $t$ is fully contained in the vector of
observations $Y_{t}:=\{y_{1},...,y_{t}\}$. The object of interest is an
estimate of the unobserved state vector $\alpha _{t},$ for which three
estimators, $\alpha _{t\mid t-1},$ $\alpha _{t\mid t}$ and $\alpha _{t\mid
n} $ are available in Bayesian framework. The first estimator is the
forecast of $\alpha _{t}$ based on information $Y_{t-1}$,%
\begin{equation*}
\alpha _{t\mid t-1}:=\mathbb{E}\left[ \alpha _{t}\mid Y_{t-1}\right] .
\end{equation*}%
Its mean square error (MSE) is defined as%
\begin{equation*}
P_{t\mid t-1}:=\mathbb{E}\left[ \left( \alpha _{t}-\alpha _{t\mid
t-1}\right) \left( \alpha _{t}-\alpha _{t\mid t-1}\right) ^{\prime }\mid
Y_{t-1}\right] .
\end{equation*}

In the linear single-regime setting with Gaussian shocks these objects and
the associated likelihood $f\left( y_{t}\mid Y_{t-1}\right) $ are computed
by the well-established technique of the standard Kalman filter (KF), which
in this case is exact and optimal (\citealp{kalman1960}). Working in a
multiple-regime environment is more challenging because of the explosive
dimensionality of the problem.

Specifically, in a multiple-regime environment, exact estimation is
infeasible because the number of histories that a Kalman-type filter needs
to take into account increases exponentially with every time period. At any
given time $t$, a multiple-regime dynamic system can be in one of $h$
possible regimes, each corresponding to a realisation of $h$ mutually
exclusive and exhaustive random events. Denote the sequence of realised
regimes from the beginning of observations up to time $t$ by $\mathcal{J}%
_{t} $:%
\begin{equation*}
\mathcal{J}_{t}=\{s_{1},s_{2},...,s_{t-1},s_{t}\} \in \mathbb{H}_{t,t}
\end{equation*}%
where $\mathbb{H}_{N,t}$ is the set of all possible histories of length $N$
that end at period $t$. There are $h^{t}$ possible mutually exclusive and
exhaustive histories up to time $t$. Using the total probability theorem,
the conditional pdf at time $t$ is obtained as a Gaussian mixture\ with the
number of terms equal to $h^{t}$: 
\begin{equation*}
f\left( y_{t+1}\mid Y_{t}\right) =\sum_{\mathcal{J}_{t}}f\left( y_{t+1}\mid 
\mathcal{J}_{t},Y_{t}\right) \Pr \left[ \mathcal{J}_{t}\mid Y_{t}\right] .
\end{equation*}

The probability of a given regime history is computed using Bayes formula:%
\begin{eqnarray*}
\Pr \left[ \mathcal{J}_{t}\mid Y_{t}\right] &=&\Pr \left[ \mathcal{J}%
_{t}\mid y_{t},Y_{t-1}\right] =\frac{f\left( y_{t}\mid \mathcal{J}%
_{t},Y_{t-1}\right) \Pr \left[ \mathcal{J}_{t}\mid Y_{t-1}\right] }{f\left(
y_{t}\mid Y_{t-1}\right) } \\
&=&\frac{f\left( y_{t}\mid \mathcal{J}_{t},Y_{t-1}\right) \Pr \left[ s_{t},%
\mathcal{J}_{t-1}\mid Y_{t-1}\right] }{f\left( y_{t}\mid Y_{t-1}\right) } \\
&=&\frac{f\left( y_{t}\mid \mathcal{J}_{t},Y_{t-1}\right) \Pr \left[
s_{t}\mid \mathcal{J}_{t-1},Y_{t-1}\right] }{f\left( y_{t}\mid
Y_{t-1}\right) }\Pr \left[ \mathcal{J}_{t-1}\mid Y_{t-1}\right]
\end{eqnarray*}%
When the regime switches have Markov property, $\Pr \left[ s_{t}\mid 
\mathcal{J}_{t-1},Y_{t-1}\right] =\func{Pr}\left[ s_{t}\mid s_{t-1}\right]
=Q(s_{t-1},s_{t})$, which simplifies the second term in the numerator.
However, conditioning on the entire past history is still needed for the
last term even if the regimes follow a Markov process.

In practice, one has to resort to some approximation. In this sense, all
practical multiple-regime filters are approximate and, therefore,
suboptimal. One popular approach involves merging two or more histories into
one. A version of this approach is well known in economic applications as
the Kim and Nelson filter. We focus on this framework and study two families
of filters with different mechanisms of approximation.

In the next section, we present the Generalised Pseudo-Bayesian (GPB)
filters, and Interacting Multiple Models (IMM) filters of arbitrary order
(length of tracked histories) $N.$

\subsection{Two Practical Families of Filters}

We begin with the GPB(N) family. It includes the Kim and Nelson filter as a
special case of GPB(2), as one can see from comparison of the expositions in %
\citet{Kim1994} and \citet{BarShalom2001}. Following commonly used
notations, the GPB(N) filter uses information from the previous $N$ periods,
including the current one. Thus, GPB(1) ignores past history and uses
current period information only, GPB(2) incorporates information from the
current period and one immediately preceding period, and so on. The IMM
algorithm is conceptually different from the GPB in the way it combines past
histories. The version of IMM developed in \citet{BlomBarshalom1988}
corresponds to IMM(1); we refer to it as canonical IMM.

One would expect that a higher $N$ leads to increased accuracy at the cost
of a larger amount of computations. We investigate relative accuracy and
speed for different $N$ within each family. Another interesting question is
whether the canonical IMM outperforms the KN filter in accuracy and speed in
a prototypical macroeconomic application.

In this section, we present the GPB(N) and the IMM(N) algorithms in turn,
using uniform notations for the entire GPB family. Where relevant, we will
remind the reader that KN is the same as GPB(2).

\subsubsection{Preliminaries}

Let $\mathcal{H}_{t}$ denote the history of regimes in $N$ consecutive
periods ending with period $t$,%
\begin{equation*}
\mathcal{H}_{t}:=\{s_{t-N+1},...,s_{t-1},s_{t}\} \in \mathbb{H}_{N,t},
\end{equation*}%
and let $\mathcal{C}_{t}$ be the `collapsed'\ history, defined as%
\begin{equation*}
\mathcal{C}_{t}:=\{s_{t-N+2},...,s_{t}\} \in \mathbb{H}_{N-1,t}.
\end{equation*}

Hence%
\begin{equation*}
\mathcal{H}_{t}=\{ \mathcal{C}_{t-1},s_{t}\}=\{s_{t-N+1},\mathcal{C}_{t}\},
\end{equation*}%
and%
\begin{eqnarray*}
\mathcal{H}_{t-1} &=&\{s_{t-N},...,s_{t-2},s_{t-1}\} \in \mathbb{H}_{N,t-1},
\\
\mathcal{H}_{t-1}\cup \mathcal{H}_{t}
&=&\{s_{t-N},...,s_{t-2},s_{t-1},s_{t}\} \in \mathbb{H}_{N+1,t}.
\end{eqnarray*}

Let%
\begin{equation*}
\mu _{t\mid t}^{(\mathcal{H}_{t})}:=\func{Pr}\left[ \mathcal{H}_{t}\mid Y_{t}%
\right]
\end{equation*}%
be the probability of realisation of a particular history $\mathcal{H}_{t}$
conditional on information at time $t$.

\subsubsection{Family of GPB Filters}

The GPB algorithm \ of order $N$, denoted GPB(N), takes into account all $%
h^{N}$ possible histories of the fixed length $N$, finishing at the current
time period. It is implemented as follows.

Define 
\begin{equation*}
\mu _{t\mid t}^{(\mathcal{C}_{t})}:=\func{Pr}\left[ \mathcal{C}_{t}\mid Y_{t}%
\right] =\sum_{s_{t-N+1}=1}^{h}\mu _{t\mid t}^{(\mathcal{H}_{t})}
\end{equation*}%
as the probability of the collapsed history, $\mathcal{C}_{t}$, conditional
on information at time $t$.

\begin{algorithm}
GPB(N) Algorithm
\end{algorithm}

\textbf{Step 0. }Start at $t=1$. Initialise $h^{N-1}$ versions of the state
vector $\alpha _{t-1\mid t-1}^{\left( \mathcal{C}_{t-1}\right) },$ the MSE
matrix $P_{t-1\mid t-1}^{\left( \mathcal{C}_{t-1}\right) },$ and
probabilities $\mu _{t-1\mid t-1}^{\left( \mathcal{C}_{t-1}\right) }$, $%
\mathcal{C}_{t-1}\in \mathbb{H}_{N-1,t-1}.$

\textbf{Step 1. }Compute $h^{N}$ standard KF forecasts to obtain%
\begin{equation*}
\begin{tabular}{cc}
\hline\hline
$\text{GPB(N) Forecast}$ & $\text{GPB(N) Update}$ \\ \hline
\multicolumn{1}{l}{$\alpha _{t\mid t-1}^{\left( \mathcal{H}_{t}\right)
}=c_{\alpha ,s_{t}}+T_{s_{t}}\alpha _{t-1\mid t-1}^{\left( \mathcal{C}%
_{t-1}\right) },$} & \multicolumn{1}{l}{$\alpha _{t\mid t}^{\left( \mathcal{H%
}_{t}\right) }=\alpha _{t\mid t-1}^{\left( \mathcal{H}_{t}\right) }+P_{t\mid
t-1}^{\left( \mathcal{H}_{t}\right) }Z_{s_{t}}^{\prime }\left[ F_{t\mid
t-1}^{\left( \mathcal{H}_{t}\right) }\right] ^{-1}v_{t\mid t-1}^{\left( 
\mathcal{H}_{t}\right) },$} \\ 
\multicolumn{1}{l}{$P_{t\mid t-1}^{\left( \mathcal{H}_{t}\right)
}=T_{s_{t}}P_{t-1\mid t-1}^{\left( \mathcal{C}_{t-1}\right)
}T_{s_{t}}^{\prime }+R_{s_{t}}R_{s_{t}}^{\prime },$} & \multicolumn{1}{l}{$%
P_{t\mid t}^{\left( \mathcal{H}_{t}\right) }=\left( I-P_{t\mid t-1}^{\left( 
\mathcal{H}_{t}\right) }Z_{s_{t}}^{\prime }\left[ F_{t\mid t-1}^{\left( 
\mathcal{H}_{t}\right) }\right] ^{-1}Z_{s_{t}}\right) P_{t\mid t-1}^{\left( 
\mathcal{H}_{t}\right) },$} \\ \hline\hline
\end{tabular}%
\end{equation*}%
where%
\begin{equation*}
\begin{array}{ccc}
v_{t\mid t-1}^{\left( \mathcal{H}_{t}\right) }=y_{t}-Z_{s_{t}}\alpha _{t\mid
t-1}^{\left( \mathcal{H}_{t}\right) }-c_{y,s_{t}}, & F_{t\mid t-1}^{\left( 
\mathcal{H}_{t}\right) }=Z_{s_{t}}P_{t\mid t-1}^{\left( \mathcal{H}%
_{t}\right) }Z_{s_{t}}^{\prime }+H_{s_{t}}, & H_{s_{t}}=g_{s_{t}}g_{s_{t}}^{%
\prime }.%
\end{array}%
\end{equation*}%
Compute the associated likelihood 
\begin{equation*}
\Lambda _{t}^{\left( \mathcal{H}_{t}\right) }=f\left( y_{t}\mid \mathcal{H}%
_{t},Y_{t-1}\right) =(2\pi )^{-p/2}\left \vert F_{t\mid t-1}^{\left( 
\mathcal{H}_{t}\right) }\right \vert ^{-1/2}\exp \left( -\frac{1}{2}v_{t\mid
t-1}^{\left( \mathcal{H}_{t}\right) \prime }\left[ F_{t\mid t-1}^{\left( 
\mathcal{H}_{t}\right) }\right] ^{-1}v_{t\mid t-1}^{\left( \mathcal{H}%
_{t}\right) }\right) .
\end{equation*}

\textbf{Step 2. }Compute probabilities $\mu _{t}^{(\mathcal{H}_{t})}$
according to\footnote{%
This procedure is a generalised version of the \citet{Hamilton1989} filter.} 
\begin{eqnarray*}
\mu _{t\mid t}^{(\mathcal{H}_{t})} &=&\Pr \left[ \mathcal{H}_{t}\mid Y_{t}%
\right] =\func{Pr}\left[ \mathcal{H}_{t}\mid Y_{t-1},y_{t}\right] \\
&=&\frac{f\left( y_{t},\mathcal{H}_{t}\mid Y_{t-1}\right) }{f\left(
y_{t}\mid Y_{t-1}\right) }=\frac{f\left( y_{t}\mid \mathcal{H}%
_{t},Y_{t-1}\right) \func{Pr}\left[ \mathcal{H}_{t}\mid Y_{t-1}\right] }{%
f\left( y_{t}\mid Y_{t-1}\right) } \\
&=&\frac{f\left( y_{t}\mid \mathcal{H}_{t},Y_{t-1}\right) \func{Pr}\left[
s_{t}\mid \mathcal{C}_{t-1},Y_{t-1}\right] }{f\left( y_{t}\mid
Y_{t-1}\right) }\Pr \left[ \mathcal{C}_{t-1}\mid Y_{t-1}\right] \\
&\simeq &\frac{\Lambda _{t}^{\left( \mathcal{H}_{t}\right) }Q\left(
s_{t-1},s_{t}\right) }{\sum_{\mathcal{H}_{t}}\Lambda _{t}^{\left( \mathcal{H}%
_{t}\right) }Q\left( s_{t-1},s_{t}\right) \mu _{t-1\mid t-1}^{(\mathcal{C}%
_{t-1})}}\mu _{t-1\mid t-1}^{(\mathcal{C}_{t-1})},
\end{eqnarray*}%
where in the last line we used approximation 
\begin{equation}
\func{Pr}\left[ s_{t}\mid \mathcal{C}_{t-1},Y_{t-1}\right] \simeq \func{Pr}%
\left[ s_{t}\mid s_{t-1}\right] =Q\left( s_{t-1},s_{t}\right) ,
\label{Markov}
\end{equation}%
and 
\begin{eqnarray}
f\left( y_{t}\mid Y_{t-1}\right) &=&\sum_{\mathcal{H}_{t}}f\left( y_{t}\mid
Y_{t-1},\mathcal{H}_{t}\right) \Pr \left[ \mathcal{H}_{t}\mid Y_{t-1}\right]
\notag \\
&=&\sum_{\mathcal{H}_{t}}f\left( y_{t}\mid Y_{t-1},\mathcal{H}_{t}\right) 
\func{Pr}\left[ s_{t}\mid \mathcal{C}_{t-1},Y_{t-1}\right] \Pr \left[ 
\mathcal{C}_{t-1}\mid Y_{t-1}\right]  \notag \\
&\simeq &\sum_{\mathcal{H}_{t}}\Lambda _{t}^{\left( \mathcal{H}_{t}\right)
}Q\left( s_{t-1},s_{t}\right) \mu _{t-1\mid t-1}^{(\mathcal{C}_{t-1})}.
\label{LGPB2}
\end{eqnarray}%
where the sum is taken over all possible realisations of $\mathcal{H}_{t}.$

\textbf{Step 3. }Next, $h^{N}$ KF outputs \ are \ merged \ into $h^{N-1}$ \
using conditional probabilities $\Pr [s_{t-N+1}|Y_{t},\mathcal{C}_{t}]$
computed as: 
\begin{equation*}
\Pr [s_{t-N+1}|Y_{t},\mathcal{C}_{t}]=\frac{\Pr \left( s_{t-N+1},\mathcal{C}%
_{t}|Y_{t}\right) }{\Pr \left( \mathcal{C}_{t}|Y_{t}\right) }=\frac{\Pr
\left( \mathcal{H}_{t}|Y_{t}\right) }{\Pr \left( \mathcal{C}%
_{t}|Y_{t}\right) }=\frac{\mu _{t\mid t}^{(\mathcal{H}_{t})}}{%
\sum_{s_{t+1-N}=1}^{h}\mu _{t\mid t}^{(\mathcal{H}_{t})}}.
\end{equation*}

Thus,%
\begin{eqnarray}
\alpha _{t\mid t}^{\left( \mathcal{C}_{t}\right) }
&=&\sum_{s_{t-N+1}=1}^{h}\Pr [s_{t-N+1}|Y_{t},\mathcal{C}_{t}]\alpha _{t\mid
t}^{\left( \mathcal{H}_{t}\right) },  \label{x_clpsd} \\
P_{t\mid t}^{\left( \mathcal{C}_{t}\right) } &=&\sum_{s_{t-N+1}=1}^{h}\Pr
[s_{t-N+1}|Y_{t},\mathcal{C}_{t}]\left \{ P_{t\mid t}^{\left( \mathcal{H}%
_{t}\right) }+\left( \alpha _{t\mid t}^{\left( \mathcal{C}_{t}\right)
}-\alpha _{t\mid t}^{\left( \mathcal{H}_{t}\right) }\right) \left( \alpha
_{t\mid t}^{\left( \mathcal{C}_{t}\right) }-\alpha _{t\mid t}^{\left( 
\mathcal{H}_{t}\right) }\right) ^{\prime }\right \}  \label{p_clpsd} \\
\mu _{t\mid t}^{\left( \mathcal{C}_{t}\right) } &=&\Pr \left[ \mathcal{C}%
_{t}\mid Y_{t}\right] =\sum_{s_{t-N+1}=1}^{h}\Pr \left[ \mathcal{H}_{t}\mid
Y_{t}\right] =\sum_{s_{t-N+1}=1}^{h}\mu _{t\mid t}^{(\mathcal{H}_{t})}.
\label{mu_clpsd}
\end{eqnarray}

Updated $\alpha _{t\mid t}^{\left( \mathcal{C}_{t}\right) },$ $P_{t\mid
t}^{\left( \mathcal{C}_{t}\right) }$ and $\mu _{t\mid t}^{\left( \mathcal{C}%
_{t}\right) }$ serve as initialisations for the next time period ($%
t=2,3,...,n$) in a recursion at Step 1.

The $t$-increment likelihood%
\begin{equation*}
L_{t}=\log f\left( y_{t}\mid Y_{t-1}\right) ,
\end{equation*}%
is computed using (\ref{LGPB2}) as part of the filter algorithm.

To \ continue \ the \ recursion \ to \ the \ end \ of \ the \ sample \ we \
only \ need \ to \ compute $\left \{ \alpha _{t\mid t}^{\left( \mathcal{C}%
_{t}\right) },P_{t\mid t}^{\left( \mathcal{C}_{t}\right) },\mu _{t\mid
t}^{\left( \mathcal{C}_{t}\right) }\right \} _{\left( \mathcal{C}_{t}\right)
\in \mathbb{H}_{N-1,t}}$\ at each step of the algorithm. These quantities
are also used to compute the state vectors and the MSE matrices $\left \{
\alpha _{t\mid t},P_{t\mid t}\right \} _{t=1:n}$ for each time $t$, and the
probability $\mu _{t}^{\left( s_{t}\right) }$ of the system being in regime $%
s_{t}$ conditional on information at time $t$: 
\begin{eqnarray}
\alpha _{t\mid t} &=&\sum_{\mathcal{C}_{t}}\Pr \left[ \mathcal{C}_{t}\mid
Y_{t}\right] \alpha _{t\mid t}^{\left( \mathcal{C}_{t}\right) }=\sum_{%
\mathcal{C}_{t}}\mu _{t\mid t}^{\left( \mathcal{C}_{t}\right) }\alpha
_{t\mid t}^{\left( \mathcal{C}_{t}\right) },  \label{fusion x} \\
P_{t\mid t} &=&\sum_{\mathcal{C}_{t}}\Pr \left[ \mathcal{C}_{t}\mid Y_{t}%
\right] \left \{ P_{t\mid t}^{\left( \mathcal{C}_{t}\right) }+\left( \alpha
_{t\mid t}-\alpha _{t\mid t}^{\left( \mathcal{C}_{t}\right) }\right) \left(
\alpha _{t\mid t}-\alpha _{t\mid t}^{\left( \mathcal{C}_{t}\right) }\right)
^{\prime }\right \}  \label{fusion p} \\
&=&\sum_{\mathcal{C}_{t}}\mu _{t\mid t}^{\left( \mathcal{C}_{t}\right)
}\left \{ P_{t\mid t}^{\left( \mathcal{C}_{t}\right) }+\left( \alpha _{t\mid
t}-\alpha _{t\mid t}^{\left( \mathcal{C}_{t}\right) }\right) \left( \alpha
_{t\mid t}-\alpha _{t\mid t}^{\left( \mathcal{C}_{t}\right) }\right)
^{\prime }\right \} ,  \notag \\
\mu _{t\mid t}^{\left( s_{t}\right) } &=&\Pr [s_{t}\mid Y_{t}]=\sum_{%
\mathcal{C}_{t-1}}\Pr [\mathcal{H}_{t}\mid Y_{t}]=\sum_{\mathcal{C}%
_{t-1}}\mu _{t\mid t}^{\left( \mathcal{H}_{t}\right) }.  \label{fusion mu}
\end{eqnarray}%
These objects can be computed outside of the recursion.

This algorithm works for any $N\geq 1$. Note that for $N=1$ we have $%
\mathcal{C}_{t}=\varnothing ,\mathcal{H}_{t}=s_{t}.$ This implies that the
GPB(1) algorithm considers possible regime histories only at their latest
instant and merges all preceding regime sequences into one, using \textit{%
common} initial conditions $\alpha _{t-1\mid t-1}$ at each time step.

\subsubsection{Family of IMM Filters}

IMM(N) maintains the dimensionality of histories at $h^{N}$ at every time
period. This is achieved by \textit{mixing} histories in every recursion of
the algorithm immediately after making a step forward in time, which would
otherwise result in an increase of dimension to $h^{N+1}.$ Effectively, at
every $t$, mixing replaces $h^{N+1}$ extended `exact' histories by $h^{N}$
reduced approximate histories weighted by probabilities of transition from
the earliest state, $s_{t-N}$, into the sequence of most recent states, $%
\{s_{t-N+1},...,s_{t}\}$. These $h^{N}$ histories are then filtered and
updated in the usual way.

A comparison of the IMM and the GPB shows that the dimension reduction is
performed at the different stages of the algorithms.\footnote{%
See Tables A1 and A2 in Appendix \ref{App Filters}.} In the IMM, mixing is
done after the state update and before the measurement update, and in the
GPB, collapsing is done after the state and measurement updates.\footnote{%
We use \textit{mixing} in the description of the IMM and \textit{collapsing}
in the description of the GPB following the convention in the literature.}

Let $\mathcal{Q}_{t|t-1}:=Q_{t|t-1}\otimes Q_{t-1|t-2}\otimes ...\otimes
Q_{t-N+1|t-N}$ denote the grand transition matrix of format $h^{N}\times
h^{N}.$

\begin{algorithm}
IMM(N) Algorithm.
\end{algorithm}

\textbf{Step 0. }Initialise $h^{N}$ versions of the state vector $\alpha
_{t-1\mid t-1}^{\left( \mathcal{H}_{t-1}\right) },$ the MSE matrix $%
P_{t-1\mid t-1}^{\left( \mathcal{H}_{t-1}\right) },$ and regime
probabilities $\mu _{t-1\mid t-1}^{(\mathcal{H}_{t-1})}$, $\mathcal{H}%
_{t-1}\in \mathbb{H}_{N,t-1}.$ Compute $\mathcal{Q}_{t|t-1}.$

\textbf{Step 1. }Compute the mixing probabilities defined as 
\begin{equation*}
\mu _{t-1\mid t-1}^{\left( \mathcal{H}_{t-1}\mid \mathcal{H}_{t}\right) }:=%
\func{Pr}\left[ \mathcal{H}_{t-1}\mid Y_{t-1},\mathcal{H}_{t}\right] .
\end{equation*}%
Note that%
\begin{eqnarray*}
\func{Pr}[\mathcal{H}_{t-1}\cup \mathcal{H}_{t}\mid Y_{t-1}] &=&\func{Pr}%
\left[ \mathcal{H}_{t-1}\mid Y_{t-1},\mathcal{H}_{t}\right] \Pr [\mathcal{H}%
_{t}\mid Y_{t-1}] \\
&=&\func{Pr}\left[ \mathcal{H}_{t}\mid \mathcal{H}_{t-1},Y_{t-1}\right] \Pr [%
\mathcal{H}_{t-1}\mid Y_{t-1}],
\end{eqnarray*}%
and 
\begin{equation*}
\Pr [\mathcal{H}_{t}\mid Y_{t-1}]=\sum_{\mathcal{H}_{t-1}}\Pr \left[ 
\mathcal{H}_{t}\mid \mathcal{H}_{t-1},Y_{t-1}\right] \Pr \left[ \mathcal{H}%
_{t-1}\mid Y_{t-1}\right] .
\end{equation*}%
Therefore, 
\begin{eqnarray*}
\mu _{t-1\mid t-1}^{\left( \mathcal{H}_{t-1}\mid \mathcal{H}_{t}\right) } &=&%
\func{Pr}\left[ \mathcal{H}_{t-1}\mid Y_{t-1},\mathcal{H}_{t}\right] =\frac{%
\func{Pr}\left[ \mathcal{H}_{t}\mid \mathcal{H}_{t-1},Y_{t-1}\right] \Pr [%
\mathcal{H}_{t-1}\mid Y_{t-1}]}{\Pr [\mathcal{H}_{t}\mid Y_{t-1}]} \\
&=&\frac{\func{Pr}\left[ \mathcal{H}_{t}\mid \mathcal{H}_{t-1},Y_{t-1}\right]
\Pr [\mathcal{H}_{t-1}\mid Y_{t-1}]}{\sum_{\mathcal{H}_{t-1}}\Pr \left[ 
\mathcal{H}_{t}\mid \mathcal{H}_{t-1},Y_{t-1}\right] \Pr \left[ \mathcal{H}%
_{t-1}\mid Y_{t-1}\right] } \\
&\simeq &\frac{\mathcal{Q}_{t|t-1}(\mathcal{H}_{t-1},\mathcal{H}_{t})\mu
_{t-1\mid t-1}^{(\mathcal{H}_{t-1})}}{\sum_{\mathcal{H}_{t-1}}\mathcal{Q}%
_{t|t-1}(\mathcal{H}_{t-1},\mathcal{H}_{t})\mu _{t-1\mid t-1}^{(\mathcal{H}%
_{t-1})}}
\end{eqnarray*}

\textbf{Step 2. }Compute the mixed\textit{\ }state vectors and MSE matrices
for each history: 
\begin{eqnarray}
\hat{\alpha}_{t-1\mid t-1}^{\left( \ast ,\mathcal{H}_{t}\right) } &=&\sum_{%
\mathcal{H}_{t-1}}\func{Pr}\left[ \mathcal{H}_{t-1}\mid Y_{t-1},\mathcal{H}%
_{t}\right] \alpha _{t-1\mid t-1}^{\left( \mathcal{H}_{t-1}\right) }=\sum_{%
\mathcal{H}_{t-1}}\mu _{t-1\mid t-1}^{\left( \mathcal{H}_{t-1}\mid \mathcal{H%
}_{t}\right) }\alpha _{t-1\mid t-1}^{\left( \mathcal{H}_{t-1}\right) },
\label{x_mixed} \\
\hat{P}_{t-1\mid t-1}^{\left( \ast ,\mathcal{H}_{t}\right) } &=&\sum_{%
\mathcal{H}_{t-1}}\func{Pr}\left[ \mathcal{H}_{t-1}\mid Y_{t-1},\mathcal{H}%
_{t}\right]  \label{P-mixed} \\
&&\times \left \{ P_{t-1\mid t-1}^{\left( s_{t-1}\right) }+\left( \alpha
_{t-1\mid t-1}^{\left( \mathcal{H}_{t-1}\right) }-\hat{\alpha}_{t-1\mid
t-1}^{\left( \ast ,\mathcal{H}_{t-1}\right) }\right) \left( \alpha _{t-1\mid
t-1}^{\left( \mathcal{H}_{t-1}\right) }-\hat{\alpha}_{t-1\mid t-1}^{\left(
\ast ,\mathcal{H}_{t-1}\right) }\right) ^{\prime }\right \}  \notag \\
&=&\sum_{\mathcal{H}_{t-1}}\mu _{t-1\mid t-1}^{\left( \mathcal{H}_{t-1}\mid 
\mathcal{H}_{t}\right) }\left \{ P_{t-1\mid t-1}^{\left( \mathcal{H}%
_{t-1}\right) }+\left( \alpha _{t-1\mid t-1}^{\left( \mathcal{H}%
_{t-1}\right) }-\hat{\alpha}_{t-1\mid t-1}^{\left( \ast ,\mathcal{H}%
_{t-1}\right) }\right) \left( \alpha _{t-1\mid t-1}^{\left( \mathcal{H}%
_{t-1}\right) }-\hat{\alpha}_{t-1\mid t-1}^{\left( \ast ,\mathcal{H}%
_{t-1}\right) }\right) ^{\prime }\right \} .  \notag
\end{eqnarray}%
Here, $\hat{\alpha}_{t-1\mid t-1}^{\left( \ast ,\mathcal{H}_{t}\right) }$ is
conditional on a particular sequence of regimes, and it is computed for all
possible sequences, or histories, in $\mathbb{H}_{N,t}.$ When computing $%
\hat{P}_{t-1\mid t-1}^{\left( \ast ,\mathcal{H}_{t}\right) }$\ in (\ref%
{P-mixed}) we take the sum over all possible sequences $\mathcal{H}_{t-1}$
of length $N$ ending at $t-1$ such that they overlap with $\mathcal{H}_{t}$
between times $t-N+1$ and $t-1.$\ Note that once states and MSEs are mixed,
the memory of $s_{t-N}$ is `cleared', so we put an asterisk in place of the
now non-existent index $s_{t-N}=\mathcal{H}_{t-1}\backslash (\mathcal{H}%
_{t}\cap \mathcal{H}_{t-1}).$ This reduces the dimensionality from $h^{N+1}$
to $h^{N}.$

\textbf{Step 3. }For each history compute the standard KF to obtain%
\begin{equation*}
\begin{tabular}{cc}
\hline\hline
$\text{IMM(N) Forecast}$ & $\text{IMM(N) Update}$ \\ \hline
\multicolumn{1}{l}{$\alpha _{t\mid t-1}^{\left( \mathcal{H}_{t}\right)
}=c_{\alpha ,s_{t}}+T_{s_{t}}\hat{\alpha}_{t-1\mid t-1}^{\left( \ast ,%
\mathcal{H}_{t}\right) },$} & \multicolumn{1}{l}{$\alpha _{t\mid t}^{\left( 
\mathcal{H}_{t}\right) }=\alpha _{t\mid t-1}^{\left( \mathcal{H}_{t}\right)
}+P_{t\mid t-1}^{\left( \mathcal{H}_{t}\right) }Z_{s_{t}}^{\prime }\left[
F_{t\mid t-1}^{\left( \mathcal{H}_{t}\right) }\right] ^{-1}v_{t\mid
-1}^{\left( \mathcal{H}_{t}\right) },$} \\ 
\multicolumn{1}{l}{$P_{t\mid t-1}^{\left( \mathcal{H}_{t}\right) }=T_{s_{t}}%
\hat{P}_{t-1\mid t-1}^{\left( \ast ,\mathcal{H}_{t}\right)
}T_{s_{t}}^{\prime }+R_{s_{t}}R_{s_{t}}^{\prime },$} & \multicolumn{1}{l}{$%
P_{t\mid t}^{\left( \mathcal{H}_{t}\right) }=\left( I-P_{t\mid t-1}^{\left( 
\mathcal{H}_{t}\right) }Z_{s_{t}}^{\prime }\left[ F_{t\mid t-1}^{\left( 
\mathcal{H}_{t}\right) }\right] ^{-1}Z_{s_{t}}\right) P_{t\mid t-1}^{\left( 
\mathcal{H}_{t}\right) },$} \\ \hline\hline
\end{tabular}%
\end{equation*}%
where%
\begin{equation*}
\begin{array}{ccc}
v_{t\mid t-1}^{\left( \mathcal{H}_{t}\right) }=y_{t}-Z_{s_{t}}\alpha _{t\mid
t-1}^{\left( \mathcal{H}_{t}\right) }-c_{y,s_{t}}, & F_{t\mid t-1}^{\left( 
\mathcal{H}_{t}\right) }=Z_{s_{t}}P_{t\mid t-1}^{\left( \mathcal{H}%
_{t}\right) }Z_{s_{t}}^{\prime }+H_{s_{t}}, & H_{s_{t}}=g_{s_{t}}g_{s_{t}}^{%
\prime }.%
\end{array}%
\end{equation*}%
Compute the associated likelihood: 
\begin{align*}
\Lambda _{t}^{\left( \mathcal{H}_{t}\right) }& =f\left( y_{t}\mid \mathcal{H}%
_{t},Y_{t-1}\right) \\
& =(2\pi )^{-p/2}\left \vert F_{t\mid t-1}^{\left( \mathcal{H}_{t}\right)
}\right \vert ^{-1/2}\exp \left( -\frac{1}{2}v_{t\mid t-1}^{\left( \mathcal{H%
}_{t}\right) \prime }\left[ F_{t\mid t-1}^{\left( \mathcal{H}_{t}\right) }%
\right] ^{-1}v_{t\mid t-1}^{\left( \mathcal{H}_{t}\right) }\right) .
\end{align*}

\textbf{Step 4. }Update the probabilities $\mu _{t\mid t}^{(\mathcal{H}%
_{t})} $, 
\begin{eqnarray*}
\mu _{t\mid t}^{(\mathcal{H}_{t})} &=&\func{Pr}\left[ \mathcal{H}_{t}\mid
Y_{t}\right] =\func{Pr}\left[ \mathcal{H}_{t}\mid Y_{t-1},y_{t}\right] \\
&=&\frac{f\left( y_{t},\mathcal{H}_{t}\mid Y_{t-1}\right) }{f\left(
y_{t}\mid Y_{t-1}\right) }=\frac{f\left( y_{t}\mid \mathcal{H}%
_{t},Y_{t-1}\right) \func{Pr}\left[ \mathcal{H}_{t}\mid Y_{t-1}\right] }{%
f\left( y_{t}\mid Y_{t-1}\right) } \\
&=&\frac{f\left( y_{t}\mid \mathcal{H}_{t},Y_{t-1}\right) \sum_{\mathcal{H}%
_{t-1}}\Pr \left[ \mathcal{H}_{t}\mid \mathcal{H}_{t-1},Y_{t-1}\right] \Pr %
\left[ \mathcal{H}_{t-1}\mid Y_{t-1}\right] }{f\left( y_{t}\mid
Y_{t-1}\right) } \\
&\simeq &\frac{\Lambda _{t}^{\left( \mathcal{H}_{t}\right) }\sum_{\mathcal{H}%
_{t-1}}\mathcal{Q}_{t|t-1}(\mathcal{H}_{t-1},\mathcal{H}_{t})\mu _{t-1\mid
t-1}^{(\mathcal{H}_{t-1})}}{\sum_{\mathcal{H}_{t}}\Lambda _{t}^{\left( 
\mathcal{H}_{t}\right) }\sum_{\mathcal{H}_{t-1}}\mathcal{Q}_{t|t-1}(\mathcal{%
H}_{t-1},\mathcal{H}_{t})\mu _{t-1\mid t-1}^{(\mathcal{H}_{t-1})}},
\end{eqnarray*}%
where in the last line we used 
\begin{eqnarray}
f\left( y_{t}\mid Y_{t-1}\right) &=&\sum_{\mathcal{H}_{t}}f\left( y_{t}\mid
Y_{t-1},\mathcal{H}_{t}\right) \Pr \left[ \mathcal{H}_{t}\mid Y_{t-1}\right]
\notag \\
&=&\sum_{\mathcal{H}_{t}}f\left( y_{t}\mid Y_{t-1},\mathcal{H}_{t}\right) 
\notag \\
&&\times \sum_{\mathcal{H}_{t-1}}\left( \Pr \left[ \mathcal{H}_{t}\mid 
\mathcal{H}_{t-1},Y_{t-1}\right] \Pr \left[ \mathcal{H}_{t-1}\mid Y_{t-1}%
\right] \right)  \notag \\
&\simeq &\sum_{s\mathcal{H}_{t}}\Lambda _{t}^{\left( \mathcal{H}_{t}\right)
}\sum_{\mathcal{H}_{t-1}}\mathcal{Q}_{t|t-1}(\mathcal{H}_{t-1},\mathcal{H}%
_{t})\mu _{t-1\mid t-1}^{(\mathcal{H}_{t-1})}  \label{LIMM1}
\end{eqnarray}

The KF outputs and the updated probabilities, $\left \{ \alpha _{t\mid
t}^{\left( \mathcal{H}_{t}\right) },P_{t\mid t}^{\left( \mathcal{H}%
_{t}\right) },\mu _{t\mid t}^{(\mathcal{H}_{t})}\right \} $, serve as
initialisations for the next time step in a recursion. The $t$-increment
likelihood is%
\begin{equation*}
L_{t}=\log f\left( y_{t}\mid Y_{t-1}\right)
\end{equation*}%
and it is computed using formula (\ref{LIMM1}) as part of the filtering
algorithm. The state vectors, the MSE matrices, and updated regime
probabilities $\left \{ \alpha _{t\mid t},P_{t\mid t},\mu _{t\mid t}^{\left(
s_{t}\right) }\right \} _{t=1:n}$ for each time $t$ are computed using the
same formulae (\ref{fusion x})-(\ref{fusion mu}) as in the GPB algorithm.

\subsection{Smoothing}

After estimating the states and regime probabilities through forward
recursion, we can improve the inference on $s_{t}$ and $\alpha _{t}$ using
the information from the entire sample. This process, known as smoothing, is
usually conducted by backward recursion. However, since smoothing algorithms
employ approximations, the smoothed estimates may not necessarily be more
precise than the filtered and updated estimates for every time period. Thus,
the overall performance of smoothing algorithms is evaluated using
simulations.

A practical algorithm for computing smoothed \textit{probabilities} for a
KN-GPB(2) filter is detailed in \citet{Kim1994}. In this paper we generalise
it for arbitrary length of history to make it applicable to filters of
higher order.

Smoothing \textit{state vectors}, even in single-regime models, can be
challenging, and is even more so in models with multiple regimes. %
\citet{Kim1994} introduces an algorithm specifically designed to work with
the KN-GPB(2) filter. However, this algorithm is computationally unstable as
it requires inverting large auxiliary matrices. Existing smoothers developed
for engineering applications typically exploit measurement errors and
require the invertibility of matrix $H_{s_{t}}=g_{s_{t}}g_{s_{t}}^{\prime }$%
, a condition often not met in economic applications.\ In this section, we
adapt a single-regime smoothing algorithm proposed by %
\citet{DurbinKoopman2012}, based on \citet{DeJong1988}, for use in a
Markov-switching multiple-regime model with arbitrary history lengths. This
adapted algorithm requires only matrix inversions that are part of the
corresponding filter and would have been computed at the filtering stage.

\subsubsection{Smoothed Probabilities}

Smoothed probabilities are computed using the total probability theorem. The
probability of history $\mathcal{H}_{t}$ conditional on the information
contained in the full sample, $Y_{n}$, can be written as%
\begin{equation*}
\mu _{t\mid n}^{(\mathcal{H}_{t})}:=\Pr [\mathcal{H}_{t}\mid
Y_{n}]=\sum_{s_{t+1}}\Pr [\mathcal{H}_{t},s_{t+1}\mid Y_{n}].
\end{equation*}%
To compute $\Pr [\mathcal{H}_{t},s_{t+1}\mid Y_{n}]$, we use the following
approximation: 
\begin{align*}
\Pr [\mathcal{H}_{t},s_{t+1}\mid Y_{n}]& =\func{Pr}\left[ s_{t+1}\mid Y_{n}%
\right] \func{Pr}\left[ \mathcal{H}_{t}\mid s_{t+1},Y_{n}\right] \\
& \simeq \func{Pr}\left[ s_{t+1}\mid Y_{n}\right] \func{Pr}\left[ \mathcal{H}%
_{t}\mid s_{t+1},Y_{t}\right] \\
& =\func{Pr}\left[ s_{t+1}\mid Y_{n}\right] \frac{\func{Pr}\left[ \mathcal{H}%
_{t},s_{t+1}\mid Y_{t}\right] }{\func{Pr}\left[ s_{t+1}\mid Y_{t}\right] } \\
& =\func{Pr}\left[ s_{t+1}\mid Y_{n}\right] \frac{\func{Pr}\left[ \mathcal{H}%
_{t}\mid Y_{t}\right] \func{Pr}\left[ s_{t+1}\mid \mathcal{H}_{t},Y_{t}%
\right] }{\sum_{\mathcal{H}_{t}}\left( \Pr \left[ s_{t+1}\mid \mathcal{H}%
_{t},Y_{t}\right] \Pr \left[ \mathcal{H}_{t}\mid Y_{t}\right] \right) } \\
& =\mu _{t+1\mid n}^{\left( s_{t+1}\right) }\frac{\mu _{t\mid t}^{(\mathcal{H%
}_{t})}Q\left( s_{t},s_{t+1}\right) }{\sum_{\mathcal{H}_{t}}Q\left(
s_{t},s_{t+1}\right) \mu _{t\mid t}^{(\mathcal{H}_{t})}}
\end{align*}%
Upon substitution, we get%
\begin{equation}
\mu _{t\mid n}^{(\mathcal{H}_{t})}\simeq \sum_{s_{t+1}}\mu _{t+1\mid
n}^{\left( s_{t+1}\right) }\frac{\mu _{t}^{(\mathcal{H}_{t})}Q\left(
s_{t},s_{t+1}\right) }{\sum_{\mathcal{H}_{t}}Q\left( s_{t},s_{t+1}\right)
\mu _{t\mid t}^{(\mathcal{H0}_{t})}}  \label{smp1}
\end{equation}

The smoothing algorithm is implemented by backward recursion as follows.

\begin{algorithm}
Smoothed Probabilities
\end{algorithm}

\textbf{Step 0.} Initialise $\mu _{n\mid n}^{\left( s_{n}\right) }=\func{Pr}%
\left[ s_{n}\mid Y_{n}\right] ,$ $s_{n}=1,...,h.$

\textbf{Step 1.} For $t=n-1$ use (\ref{smp1}) to compute the smoothed
probability of $\mu _{t\mid n}^{(\mathcal{H}_{t})}$ for history $\mathcal{H}%
_{t}$ .

\textbf{Step 2.} Compute smoothed probabilities of each regime:%
\begin{equation*}
\mu _{t\mid n}^{\left( s_{t}\right) }=\Pr \left[ s_{t}\mid Y_{n}\right]
=\sum_{\mathcal{C}_{t-1}}\Pr \left[ \mathcal{H}_{t}\mid Y_{n}\right] =\sum_{%
\mathcal{C}_{t-1}}\mu _{t\mid n}^{(\mathcal{H}_{t})}
\end{equation*}%
Use $\mu _{t\mid n}^{\left( s_{t}\right) }$ to initialise the algorithm for $%
t=n-2.$

\subsubsection{Smoothed Variables}

The smoother is based on the properties of the joint Gaussian distribution
of the forecast errors of the vector of latent state variables and the
estimation errors of the vector of observations produced by the filter.

By definition, smoothed state vectors and MSE matrices are: 
\begin{eqnarray}
\alpha _{t\mid n}^{\left( \mathcal{H}_{t}\right) } &=&\mathbb{E}\left[
\alpha _{t}^{\left( \mathcal{H}_{t}\right) }\mid Y_{n},\mathcal{H}_{t}\right]
,  \label{as} \\
P_{t\mid n}^{\left( \mathcal{H}_{t}\right) } &=&\mathbb{E}\left[ \left(
\alpha _{t}^{\left( \mathcal{H}_{t}\right) }-\alpha _{t\mid t-1}^{\left( 
\mathcal{H}_{t}\right) }\right) \left( \alpha _{t}^{\left( \mathcal{H}%
_{t}\right) }-\alpha _{t\mid t-1}^{\left( \mathcal{H}_{t}\right) }\right)
^{\prime }\mid Y_{n},\mathcal{H}_{t}\right] .  \label{ps}
\end{eqnarray}

Define the forecast error of the state vector at time $t$ with history $%
\mathcal{H}_{t}$ as 
\begin{equation}
\xi _{t\mid t-1}^{\left( \mathcal{H}_{t}\right) }:=\alpha _{t}-\alpha
_{t\mid t-1}^{\left( \mathcal{H}_{t}\right) }  \label{xi_def}
\end{equation}%
\qquad Then,%
\begin{equation}
P_{t\mid t-1}^{\left( \mathcal{H}_{t}\right) }=\mathbb{E}\left[ \xi _{t\mid
t-1}^{\left( \mathcal{H}_{t}\right) }\xi _{t\mid t-1}^{\left( \mathcal{H}%
_{t}\right) \prime }\mid Y_{t-1},\mathcal{H}_{t}\right] .
\end{equation}

Define the Kalman gain matrix:%
\begin{equation}
K_{t\mid t-1}^{\left( \mathcal{H}_{t}\right) }=P_{t\mid t-1}^{\left( 
\mathcal{H}_{t}\right) }Z_{s_{t}}^{\prime }\left[ F_{t\mid t-1}^{\left( 
\mathcal{H}_{t}\right) }\right] ^{-1}.  \label{Kgain}
\end{equation}

To calculate $\alpha _{t\mid n}^{\left( \mathcal{H}_{t}\right) }$ defined in
(\ref{as}), we split the history into two components at $t-1$ and use the
formula for the conditional mean of multivariate Gaussian distribution:%
\begin{eqnarray*}
\alpha _{t\mid n}^{\left( \mathcal{H}_{t}\right) } &=&\mathbb{E}\left[
\alpha _{t}\mid \mathcal{H}_{t},Y_{n}\right] =\mathbb{E}\left[ \alpha
_{t}\mid \mathcal{H}_{t},Y_{t-1},\{v_{k\mid k-1}\}_{k=t:n}\right] \\
&=&\alpha _{t\mid t-1}^{\left( \mathcal{H}_{t}\right) }+\sum_{k=t}^{n}%
\mathbb{E}\left[ \alpha _{t}v_{k\mid k-1}^{\prime }\mid Y_{t-1},\mathcal{H}%
_{t}\right] F_{k\mid k-1}^{-1}v_{k\mid k-1} \\
&=&\alpha _{t\mid t-1}^{\left( \mathcal{H}_{t}\right) }+\sum_{k=t}^{n}%
\mathbb{E}\left[ \left( \xi _{t\mid t-1}^{\left( \mathcal{H}_{t}\right)
}+\alpha _{t\mid t-1}^{\left( \mathcal{H}_{t}\right) }\right) \left( \xi
_{k\mid k-1}^{\prime }Z_{s_{k}}^{\prime }+\left[ g_{s_{k}}\varepsilon _{k}%
\right] ^{\prime }\right) \mid Y_{t-1},\mathcal{H}_{t}\right] F_{k\mid
k-1}^{-1}v_{k\mid k-1}
\end{eqnarray*}%
where we used%
\begin{equation*}
v_{t\mid t-1}^{\left( \mathcal{H}_{t}\right) }=y_{t}-Z_{s_{t}}\alpha _{t\mid
t-1}^{\left( \mathcal{H}_{t}\right) }-c_{y,s_{t}}=y_{t}-Z_{s_{t}}\left(
\alpha _{t}-\xi _{t\mid t-1}^{\left( \mathcal{H}_{t}\right) }\right)
-c_{y,s_{t}}=g_{s_{t}}\varepsilon _{t}+Z_{s_{t}}\xi _{t\mid t-1}^{\left( 
\mathcal{H}_{t}\right) }.
\end{equation*}%
So, finally, 
\begin{equation}
\alpha _{t\mid n}^{\left( \mathcal{H}_{t}\right) }=\alpha _{t\mid
t-1}^{\left( \mathcal{H}_{t}\right) }+\sum_{k=t}^{n}\mathbb{E}\left[ \xi
_{t\mid t-1}^{\left( \mathcal{H}_{t}\right) }\xi _{k\mid k-1}^{\prime }\mid
Y_{t-1},\mathcal{H}_{t}\right] Z_{s_{k}}^{\prime }F_{k\mid k-1}^{-1}v_{k\mid
k-1}.  \label{a_t_n}
\end{equation}

Similarly, the formula for conditional variance of multivariate Gaussian
distribution applied to (\ref{ps}) yields:%
\begin{eqnarray*}
P_{t\mid n}^{\left( \mathcal{H}_{t}\right) } &=&P_{t\mid t-1}^{\left( 
\mathcal{H}_{t}\right) }-\sum_{k=t}^{n}\mathbb{E}\left[ \alpha _{t}v_{k\mid
k-1}^{\prime }\mid Y_{t-1},\mathcal{H}_{t}\right] F_{k\mid k-1}^{-1}\mathbb{E%
}\left[ v_{k\mid k-1}\alpha _{t}^{\prime }\mid Y_{t-1},\mathcal{H}_{t}\right]
\\
&=&P_{t\mid t-1}^{\left( \mathcal{H}_{t}\right) }-\sum_{k=t}^{n}\mathbb{E}%
\left[ \left( \xi _{t\mid t-1}^{\left( \mathcal{H}_{t}\right) }+\alpha
_{t\mid t-1}^{\left( \mathcal{H}_{t}\right) }\right) \left( \xi _{k\mid
k-1}^{\prime }Z_{s_{k}}^{\prime }+\left[ g_{s_{k}}\varepsilon _{k}\right]
^{\prime }\right) \mid Y_{t-1},\mathcal{H}_{t}\right] F_{k\mid k-1}^{-1} \\
&&\times \mathbb{E}\left[ \left( Z_{s_{k}}\xi _{k\mid k-1}+\left[
g_{s_{k}}\varepsilon _{k}\right] \right) \left( \xi _{t\mid t-1}^{\left( 
\mathcal{H}_{t}\right) \prime }+\alpha _{t\mid t-1}^{\left( \mathcal{H}%
_{t}\right) \prime }\right) \mid Y_{t-1},\mathcal{H}_{t}\right] ,
\end{eqnarray*}%
so that%
\begin{equation}
P_{t\mid n}^{\left( \mathcal{H}_{t}\right) }=P_{t\mid t-1}^{\left( \mathcal{H%
}_{t}\right) }-\sum_{k=t}^{n}\mathbb{E}\left[ \xi _{t\mid t-1}^{\left( 
\mathcal{H}_{t}\right) }\xi _{k\mid k-1}^{\prime }\mid Y_{t-1},\mathcal{H}%
_{t}\right] Z_{s_{k}}^{\prime }F_{k\mid k-1}^{-1}Z_{s_{k}}\mathbb{E}\left[
\xi _{k\mid k-1}\xi _{t\mid t-1}^{\left( \mathcal{H}_{t}\right) \prime }\mid
Y_{t-1}\right]  \label{P_t_n}
\end{equation}

In expressions (\ref{a_t_n}) and (\ref{P_t_n}) we do not specify the \textit{%
future} regime sequences starting from $s_{t}$ under the summation. We
introduce them in the calculations of expectations $\mathbb{E}\left[ \cdot
\mid \cdot \right] $ for every step going backward, as shown later.

For now, we will need the following recursion for $\xi _{t\mid t-1}^{\left( 
\mathcal{H}_{t}\right) }.$ The recursion is slightly different for the two
families of filters.

\begin{lemma}
\label{Lemma}1. For GPB(N) filter%
\begin{equation}
\xi _{t\mid t-1}^{\left( \mathcal{H}_{t}\right)
}=T_{s_{t}}\sum_{s_{t-N}=1}^{h}\Pr [s_{t-N}|Y_{t-1},\mathcal{C}_{t-1}]\left(
I-K_{t-1\mid t-2}^{\left( \mathcal{H}_{t-1}\right) }Z_{s_{t-1}}\right) \xi
_{t-1\mid t-2}^{\left( \mathcal{H}_{t-1}\right) }+\omega _{t-1},
\label{LemmaGPB}
\end{equation}%
where 
\begin{equation*}
\omega _{t-1}=R_{s_{t}}\eta _{t}-T_{s_{t}}\sum_{s_{t-N}=1}^{h}\Pr
[s_{t-N}|Y_{t-1},\mathcal{C}_{t-1}]K_{t-1\mid t-2}^{\left( \mathcal{H}%
_{t-1}\right) }g_{s_{t-1}}\varepsilon _{t-1}.
\end{equation*}

2. For IMM filter%
\begin{equation}
\xi _{t\mid t-1}^{\left( \mathcal{H}_{t}\right) }=T_{s_{t}}\sum_{\mathcal{H}%
_{t-1}}\func{Pr}\left[ \mathcal{H}_{t-1}\mid Y_{t-1},\mathcal{H}_{t}\right]
\left( I-K_{t-1\mid t-2}^{\left( \mathcal{H}_{t-1}\right)
}Z_{s_{t-1}}\right) \xi _{t-1\mid t-2}^{\left( \mathcal{H}_{t-1}\right)
}+\omega _{t-1},  \label{LemmaIMM}
\end{equation}%
where 
\begin{equation*}
\omega _{t-1}=R_{s_{t}}\eta _{t}-T_{s_{t}}\sum_{\mathcal{H}_{t-1}}\func{Pr}%
\left[ \mathcal{H}_{t-1}\mid Y_{t-1},\mathcal{H}_{t}\right] K_{t-1\mid
t-2}^{\left( s_{t-1}\right) }g_{s_{t-1}}\varepsilon _{t-1}.
\end{equation*}
\end{lemma}

\begin{proof}
For GPB(N) we have 
\begin{equation*}
\xi _{t\mid t-1}^{\left( \mathcal{H}_{t}\right) }=\alpha _{t}-\alpha _{t\mid
t-1}^{\left( \mathcal{H}_{t}\right) }=T_{s_{t}}\left( \alpha _{t-1}-\alpha
_{t-1\mid t-1}^{\left( \mathcal{C}_{t-1}\right) }\right) +R_{s_{t}}\eta _{t},
\end{equation*}%
and, using $\alpha _{t-1\mid t-1}^{\left( \mathcal{C}_{t-1}\right) }$ from
equation (\ref{x_clpsd}), 
\begin{equation}
\xi _{t\mid t-1}^{\left( \mathcal{H}_{t}\right) }=T_{s_{t}}\left( \alpha
_{t-1}-\sum_{s_{t-N}=1}^{h}\Pr [s_{t-N}\mid Y_{t-1},\mathcal{C}_{t-1}]\alpha
_{t-1\mid t-1}^{\left( \mathcal{H}_{t-1}\right) }\right) +R_{s_{t}}\eta _{t}.
\label{GPBxi}
\end{equation}%
Similarly, for IMM(N) we have 
\begin{equation*}
\xi _{t\mid t-1}^{\left( \mathcal{H}_{t}\right) }=T_{s_{t}}\left( \alpha
_{t-1}-\hat{\alpha}_{t-1\mid t-1}^{\left( \mathcal{H}_{t-1}\mid \mathcal{H}%
_{t}\right) }\right) +R_{s_{t}}\eta _{t}
\end{equation*}%
and using $\hat{\alpha}_{t-1\mid t-1}^{\left( \mathcal{H}_{t-1}\mid \mathcal{%
H}_{t}\right) }$ from (\ref{x_mixed}), 
\begin{equation}
\xi _{t\mid t-1}^{\left( \mathcal{H}_{t}\right) }=T_{s_{t}}\left( \alpha
_{t-1}-\sum_{\mathcal{H}_{t-1}}\func{Pr}\left[ \mathcal{H}_{t-1}\mid Y_{t-1},%
\mathcal{H}_{t}\right] \alpha _{t-1\mid t-1}^{\left( \mathcal{H}%
_{t-1}\right) }\right) +R_{s_{t}}\eta _{t}.  \label{IMMxi}
\end{equation}%
Denote 
\begin{equation*}
M\left( \psi \right) =\left \{ 
\begin{array}{ccc}
\Pr [\psi \mid Y_{t-1},\mathcal{C}_{t-1}], & \psi =s_{t-N}, & \text{for
GPB(N),} \\ 
\func{Pr}\left[ \psi \mid Y_{t-1},\mathcal{H}_{t}\right] , & \psi =\mathcal{H%
}_{t-1}, & \text{for IMM(N).}%
\end{array}%
\right.
\end{equation*}%
Then expressions (\ref{GPBxi}) and (\ref{IMMxi}) can be written in the same
form,%
\begin{equation*}
\xi _{t\mid t-1}^{\left( \mathcal{H}_{t}\right) }=T_{s_{t}}\left( \alpha
_{t-1}-\sum_{\psi }M\left( \psi \right) \alpha _{t-1\mid t-1}^{\left( 
\mathcal{H}_{t-1}\right) }\right) +R_{s_{t}}\eta _{t},
\end{equation*}%
and the rest of the proof is identical for both families of filters.\newline
Use the KF update for $\alpha _{t-1\mid t-1}^{\left( \mathcal{H}%
_{t-1}\right) }$ and the definition of Kalman gain (\ref{Kgain}) for $%
K_{t-1\mid t-2}^{\left( \mathcal{H}_{t-1}\right) }$ to rewrite the last
expression as:%
\begin{eqnarray*}
\xi _{t\mid t-1}^{\left( \mathcal{H}_{t}\right) } &=&T_{s_{t}}\left( \alpha
_{t-1}-\sum_{\psi }M\left( \psi \right) \right. \\
&&\left. \times \left( \alpha _{t-1\mid t-2}^{\left( \mathcal{H}%
_{t-1}\right) }+P_{t-1\mid t-2}^{\left( \mathcal{H}_{t-1}\right)
}Z_{s_{t-1}}^{\prime }\left[ F_{t-1\mid t-2}^{\left( \mathcal{H}%
_{t-1}\right) }\right] ^{-1}v_{t-1\mid t-2}^{\left( \mathcal{H}_{t-1}\right)
}\right) \right) +R_{s_{t}}\eta _{t} \\
&=&T_{s_{t}}\sum_{\psi }M\left( \psi \right) \left( \alpha _{t-1}-\alpha
_{t-1\mid t-2}^{\left( \mathcal{H}_{t-1}\right) }\right) \\
&&-T_{s_{t}}\sum_{\psi }M\left( \psi \right) K_{t-1\mid t-2}^{\left( 
\mathcal{H}_{t-1}\right) }v_{t-1\mid t-2}^{\left( \mathcal{H}_{t-1}\right)
}+R_{s_{t}}\eta _{t}
\end{eqnarray*}%
Next, use the KF output for $v_{t-1\mid t-2}^{\left( \mathcal{H}%
_{t-1}\right) }$ along with the definition (\ref{xi_def}) for $\xi _{t-1\mid
t-2}^{\left( \mathcal{H}_{t-1}\right) }$ and equation (\ref{me}) for $%
y_{t-1} $ to obtain the recursions in Lemma \ref{Lemma}:%
\begin{eqnarray*}
\xi _{t\mid t-1}^{\left( \mathcal{H}_{N,t}\right) } &=&T_{s_{t}}\sum_{\psi
}M\left( \psi \right) \left( \alpha _{t-1}-\alpha _{t-1\mid t-2}^{\left( 
\mathcal{H}_{t-1}\right) }\right) \\
&&-T_{s_{t}}\sum_{\psi }M\left( \psi \right) K_{t-1\mid t-2}^{\left( 
\mathcal{H}_{t-1}\right) }\left( y_{t-1}-Z_{s_{t-1}}\alpha _{t-1\mid
t-2}^{\left( \mathcal{H}_{t-1}\right) }-c_{y,s_{t-1}}\right) +R_{s_{t}}\eta
_{t}
\end{eqnarray*}%
\begin{eqnarray*}
&=&T_{s_{t}}\sum_{\psi }M\left( \psi \right) \xi _{t-1\mid t-2}^{\left( 
\mathcal{H}_{t-1}\right) } \\
&&-T_{s_{t}}\sum_{\psi }M\left( \psi \right) K_{t-1\mid t-2}^{\left( 
\mathcal{H}_{t-1}\right) }Z_{s_{t-1}}\xi _{t-1\mid t-2}^{\left( \mathcal{H}%
_{t-1}\right) } \\
&&+R_{s_{t}}\eta _{t}-T_{s_{t}}\sum_{\psi }M\left( \psi \right) K_{t-1\mid
t-2}^{\left( \mathcal{H}_{t-1}\right) }g_{s_{t-1}}\varepsilon _{t-1} \\
&=&T_{s_{t}}\sum_{\psi }M\left( \psi \right) \left( I-K_{t-1\mid
t-2}^{\left( \mathcal{H}_{t-1}\right) }Z_{s_{t-1}}\right) \xi _{t-1\mid
t-2}^{\left( \mathcal{H}_{t-1}\right) }+\omega _{t-1}^{\left( \mathcal{H}%
_{t-1}\right) }
\end{eqnarray*}%
where we used the notation%
\begin{equation*}
\omega _{t-1}=R_{s_{t}}\eta _{t}-T_{s_{t}}\sum_{\psi }M\left( \psi \right)
K_{t-1\mid t-2}^{\left( \mathcal{H}_{t-1}\right) }g_{s_{t-1}}\varepsilon
_{t-1}.
\end{equation*}
\end{proof}

\bigskip Note that in this derivation for GPB(N) the sum is taken over all
possible regimes at time $t-N$ when $s_{t-N}$ is unknown. If the regime $%
s_{t-N}$ is known, this recursion is given by%
\begin{equation}
\xi _{t\mid t-1}^{\left( s_{t-N},\mathcal{H}_{t}\right) }=T_{s_{t}}\left(
I-K_{t-1\mid t-2}^{\left( s_{t-N},\mathcal{H}_{t-1}\right)
}Z_{s_{t-1}}\right) \xi _{t-1\mid t-2}^{\left( s_{t-N},\mathcal{H}%
_{t-1}\right) }+\omega _{t-1}.  \label{xi_m_GPB}
\end{equation}%
Similarly, for IMM, when $\mathcal{H}_{t-1}=\mathcal{\tilde{H}}_{t-1}$ is
known, then%
\begin{equation}
\xi _{t\mid t-1}^{\left( \mathcal{\tilde{H}}_{t-1},\mathcal{H}_{t}\right)
}=T_{s_{t}}\left( I-K_{t-1\mid t-2}^{\left( \mathcal{\tilde{H}}_{t-1},%
\mathcal{H}_{t-1}\right) }Z_{s_{t-1}}\right) \xi _{t-1\mid t-2}^{\left( 
\mathcal{\tilde{H}}_{t-1},\mathcal{H}_{t-1}\right) }+\omega _{t-1}.
\label{xi_m_IMM}
\end{equation}

The remaining derivations are identical for GPB and IMM.

We apply formulas (\ref{a_t_n}) and (\ref{P_t_n}) recursively, starting from
the observation at the final period, $n$, in regime $s_{n}$: 
\begin{eqnarray*}
\alpha _{n\mid n}^{\left( \mathcal{H}_{n}\right) } &=&\alpha _{n\mid
n-1}^{\left( \mathcal{H}_{n}\right) }+\mathbb{E}\left[ \xi _{n\mid
n-1}^{\left( \mathcal{H}_{n}\right) }\xi _{n\mid n-1}^{\left( \mathcal{H}%
_{n}\right) \prime }\mid Y_{n-1},\mathcal{H}_{n}\right] Z_{s_{n}}^{\prime }%
\left[ F_{n\mid n-1}^{\left( \mathcal{H}_{n}\right) }\right] ^{-1}v_{n\mid
n-1}^{\left( \mathcal{H}_{n}\right) } \\
&=&\alpha _{n\mid n-1}^{\left( \mathcal{H}_{n}\right) }+P_{n\mid
n-1}^{\left( \mathcal{H}_{n}\right) }Z_{s_{n}}^{\prime }\left[ F_{n\mid
n-1}^{\left( \mathcal{H}_{n}\right) }\right] ^{-1}v_{n\mid n-1}^{\left( 
\mathcal{H}_{n}\right) }=\alpha _{n\mid n-1}^{\left( \mathcal{H}_{n}\right)
}+P_{n\mid n-1}^{\left( \mathcal{H}_{n}\right) }r_{n\mid n-1}^{\left( 
\mathcal{H}_{n}\right) }
\end{eqnarray*}%
where%
\begin{equation*}
r_{n\mid n-1}^{\left( \mathcal{H}_{n}\right) }=Z_{s_{n}}^{\prime }\left[
F_{n\mid n-1}^{\left( \mathcal{H}_{n}\right) }\right] ^{-1}v_{n\mid
n-1}^{\left( \mathcal{H}_{n}\right) },
\end{equation*}%
and%
\begin{eqnarray*}
P_{n\mid n}^{\left( \mathcal{H}_{n}\right) } &=&P_{n\mid n-1}^{\left( 
\mathcal{H}_{n}\right) }-\mathbb{E}\left[ \xi _{n\mid n-1}^{\left( \mathcal{H%
}_{n}\right) }\xi _{n\mid n-1}^{\left( \mathcal{H}_{n}\right) \prime }\mid
Y_{n-1},\mathcal{H}_{n}\right] \\
&&\times Z_{s_{n}}^{\prime }\left[ F_{n\mid n-1}^{\left( \mathcal{H}%
_{n}\right) }\right] ^{-1}Z_{s_{n}}\mathbb{E}\left[ \xi _{n\mid n-1}\xi
_{n\mid n-1}^{\left( \mathcal{H}_{n}\right) \prime }\mid Y_{n-1},\mathcal{H}%
_{n}\right] \\
&=&P_{n\mid n-1}^{\left( \mathcal{H}_{n}\right) }-P_{n\mid n-1}^{\left( 
\mathcal{H}_{n}\right) }Z_{s_{n}}^{\prime }\left[ F_{n\mid n-1}^{\left( 
\mathcal{H}_{n}\right) }\right] ^{-1}Z_{s_{n}}P_{n\mid n-1}^{\left( \mathcal{%
H}_{n}\right) } \\
&=&P_{n\mid n-1}^{\left( \mathcal{H}_{n}\right) }-P_{n\mid n-1}^{\left( 
\mathcal{H}_{n}\right) }N_{n\mid n-1}^{\left( \mathcal{H}_{n}\right)
}P_{n\mid n-1}^{\left( \mathcal{H}_{n}\right) },
\end{eqnarray*}%
where%
\begin{equation*}
N_{n\mid n-1}^{\left( \mathcal{H}_{n}\right) }=Z_{s_{n}}^{\prime }\left[
F_{n\mid n-1}^{\left( \mathcal{H}_{n}\right) }\right] ^{-1}Z_{s_{n}}.
\end{equation*}

Next, we move one step back to $t=n-1$.%
\begin{eqnarray}
\alpha _{n-1\mid n}^{\left( \mathcal{H}_{n-1}\right) } &=&\alpha _{n-1\mid
n-2}^{\left( \mathcal{H}_{n-1}\right) }+\sum_{k=n-1}^{n}\mathbb{E}\left[ \xi
_{n-1\mid n-2}^{\left( \mathcal{H}_{n-1}\right) }\xi _{k\mid k-1}^{\prime
}\mid Y_{n-2},\mathcal{H}_{n-1}\right] Z_{s_{k}}^{\prime }\left[ F_{k\mid
k-1}\right] ^{-1}v_{k\mid k-1}  \notag \\
&=&\alpha _{n-1\mid n-2}^{\left( \mathcal{H}_{n-1}\right) }+P_{n-1\mid
n-2}^{\left( \mathcal{H}_{n-1}\right) }Z_{s_{n-1}}^{\prime }\left[
F_{n-1\mid n-2}^{\left( \mathcal{H}_{n-1}\right) }\right] ^{-1}v_{n-1\mid
n-2}^{\left( \mathcal{H}_{n-1}\right) }  \notag \\
&&+\sum_{s_{n}=1}^{h}\Pr \left[ s_{n}\mid Y_{n-2},\mathcal{H}_{n-1}\right] 
\notag \\
&&\times \mathbb{E}\left[ \xi _{n-1\mid n-2}^{\left( \mathcal{H}%
_{n-1}\right) }\xi _{n\mid n-1}^{\left( s_{n-N},\mathcal{H}_{n-1}\right)
\prime }\mid Y_{n-2},\mathcal{H}_{n-1},s_{n}\right] Z_{s_{n}}^{\prime }\left[
F_{n\mid n-1}^{\left( \mathcal{H}_{n}\right) }\right] ^{-1}v_{n\mid
n-1}^{\left( \mathcal{H}_{n}\right) }  \notag \\
&=&\alpha _{n-1\mid n-2}^{\left( \mathcal{H}_{n-1}\right) }+P_{n-1\mid
n-2}^{\left( \mathcal{H}_{n-1}\right) }Z_{s_{n-1}}^{\prime }\left[
F_{n-1\mid n-2}^{\left( \mathcal{H}_{n-1}\right) }\right] ^{-1}v_{n-1\mid
n-2}^{\left( \mathcal{H}_{n-1}\right) }  \notag \\
&&+\sum_{s_{n}=1}^{h}\Pr \left[ s_{n}\mid Y_{n-2},\mathcal{H}_{n-1}\right] 
\mathbb{E}\left[ \xi _{n-1\mid n-2}^{\left( \mathcal{H}_{n-1}\right) }\xi
_{n-1\mid n-2}^{\left( \mathcal{H}_{n-1}\right) \prime }\mid Y_{n-2},%
\mathcal{H}_{n-1}\right]  \notag \\
&&\times \left( I-Z_{s_{n-1}}^{\prime }K_{n-1\mid n-2}^{\left( \mathcal{H}%
_{n-1}\right) \prime }\right) T_{s_{t}}^{\prime }Z_{s_{n}}^{\prime }\left[
F_{n\mid n-1}^{\left( \mathcal{H}_{n}\right) }\right] ^{-1}v_{n\mid
n-1}^{\left( \mathcal{H}_{n}\right) }  \notag \\
&=&\alpha _{n-1\mid n-2}^{\left( \mathcal{H}_{n-1}\right) }+P_{n-1\mid
n-2}^{\left( \mathcal{H}_{n-1}\right) }Z_{s_{n-1}}^{\prime }\left[
F_{n-1\mid n-2}^{\left( \mathcal{H}_{n-1}\right) }\right] ^{-1}v_{n-1\mid
n-2}^{\left( \mathcal{H}_{n-1}\right) }  \label{anm1} \\
&&+P_{n-1\mid n-2}^{\left( \mathcal{H}_{n-1}\right) }\sum_{s_{n}=1}^{h}\Pr %
\left[ s_{n}\mid Y_{n-2},\mathcal{H}_{n-1}\right]  \notag \\
&&\times \left( I-Z_{s_{n-1}}^{\prime }K_{n-1\mid n-2}^{\left( \mathcal{H}%
_{n-1}\right) \prime }\right) T_{s_{t}}^{\prime }Z_{s_{n}}^{\prime }\left[
F_{n\mid n-1}^{\left( \mathcal{H}_{n}\right) }\right] ^{-1}v_{n\mid
n-1}^{\left( \mathcal{H}_{n}\right) }.  \notag
\end{eqnarray}%
Similar computations yield%
\begin{eqnarray}
P_{n-1\mid n}^{\left( \mathcal{H}_{n-1}\right) } &=&P_{n-1\mid n-2}^{\left( 
\mathcal{H}_{n-1}\right) }-\sum_{k=n-1}^{n}\mathbb{E}\left[ \xi _{n-1\mid
n-2}^{\left( \mathcal{H}_{n-1}\right) }\xi _{k\mid k-1}^{\prime }\mid
Y_{n-2},\mathcal{H}_{n-1}\right] Z_{s_{k}}^{\prime }  \notag \\
&&\times \left[ F_{n-1\mid n-2}^{\left( \mathcal{H}_{n-1}\right) }\right]
^{-1}Z_{s_{k}}\mathbb{E}\left[ \xi _{k\mid k-1}\xi _{n-1\mid n-2}^{\left( 
\mathcal{H}_{n-1}\right) \prime }\mid Y_{n-2},\mathcal{H}_{n-1}\right] 
\notag \\
&=&P_{n-1\mid n-2}^{\left( \mathcal{H}_{n-1}\right) }-P_{n-1\mid
n-2}^{\left( \mathcal{H}_{n-1}\right) }Z_{s_{n-1}}^{\prime }\left[
F_{n-1\mid n-2}^{\left( \mathcal{H}_{n-1}\right) }\right]
^{-1}Z_{s_{n-1}}P_{n-1\mid n-2}^{\left( \mathcal{H}_{n-1}\right) }
\label{pnm1} \\
&&-P_{n-1|n-2}^{\left( \mathcal{H}_{n-1}\right) }\sum_{s_{n}=1}^{h}\Pr \left[
s_{n}\mid Y_{n-2},\mathcal{H}_{n-1}\right] \left( I-Z_{s_{n-1}}^{\prime
}K_{n-1\mid n-2}^{\left( \mathcal{H}_{n-1}\right) \prime }\right)  \notag \\
&&\times T_{s_{t}}^{\prime }Z_{s_{n}}^{\prime }\left[ F_{n\mid n-1}^{\left( 
\mathcal{H}_{n}\right) }\right] ^{-1}Z_{s_{n-1}}T_{s_{t}}\left( I-K_{n-1\mid
n-2}^{\left( \mathcal{H}_{n-1}\right) }Z_{s_{n-1}}\right) P_{n-1\mid
n-2}^{\left( \mathcal{H}_{n-1}\right) }.  \notag
\end{eqnarray}

In \ \ these \ \ derivations \ \ \ we \ \ used $P_{n-1\mid n-2}^{\left( 
\mathcal{H}_{n-1}\right) }$ $=$ $\mathbb{E}\left[ \xi _{n-1\mid n-2}^{\left( 
\mathcal{H}_{n-1}\right) }\xi _{n-1\mid n-2}^{\left( \mathcal{H}%
_{n-1}\right) \prime }\mid Y_{n-2},\mathcal{H}_{n-1},s_{n}\right] $ $=$ $%
\mathbb{E}\left[ \xi _{n-1\mid n-2}^{\left( \mathcal{H}_{n-1}\right) }\xi
_{n-1\mid n-2}^{\left( \mathcal{H}_{n-1}\right) \prime }\mid Y_{n-2},%
\mathcal{H}_{n-1}\right] $ as conditioning on $s_{n}$ becomes irrelevant.

Equations (\ref{anm1}) and (\ref{pnm1}) can be written as: 
\begin{eqnarray*}
\alpha _{n-1\mid n}^{\left( \mathcal{H}_{n-1}\right) } &=&\alpha _{n-1\mid
n-2}^{\left( \mathcal{H}_{n-1}\right) }+P_{n-1\mid n-2}^{\left( \mathcal{H}%
_{n-1}\right) }r_{n-1\mid n-2}^{\left( \mathcal{H}_{n-1}\right) }, \\
P_{n-1\mid n}^{\left( \mathcal{H}_{n-1}\right) } &=&P_{n-1\mid n-2}^{\left( 
\mathcal{H}_{n-1}\right) }-P_{n-1\mid n-2}^{\left( \mathcal{H}_{n-1}\right)
}N_{n-1\mid n-2}^{\left( \mathcal{H}_{n}\right) }P_{n-1\mid n-2}^{\left( 
\mathcal{H}_{n-1}\right) },
\end{eqnarray*}%
where, using approximation $\Pr \left[ s_{n}\mid Y_{n-2},\mathcal{H}_{n-1}%
\right] \simeq Q\left( s_{n-1},s_{n}\right) ,$ we express $r_{n-1\mid
n-2}^{\left( \mathcal{H}_{n-1}\right) }$ and $N_{n-1\mid n-2}^{\left( 
\mathcal{H}_{n}\right) }$ recursively: 
\begin{eqnarray*}
r_{n-1\mid n-2}^{\left( \mathcal{H}_{n-1}\right) } &=&Z_{s_{n-1}}^{\prime }%
\left[ F_{n-1\mid n-2}^{\left( \mathcal{H}_{n-1}\right) }\right]
^{-1}v_{n-1\mid n-2}^{\left( \mathcal{H}_{n-1}\right)
}+\sum_{s_{n}=1}^{h}Q\left( s_{n-1},s_{n}\right) L_{n,n-1}^{\left( \mathcal{H%
}_{n-1}\right) \prime }r_{n\mid n-1}^{\left( \mathcal{H}_{n}\right) }, \\
N_{n-1\mid n-2}^{\left( \mathcal{H}_{n-1}\right) } &=&Z_{s_{n-1}}^{\prime }%
\left[ F_{n-1\mid n-2}^{\left( \mathcal{H}_{n-1}\right) }\right]
^{-1}Z_{s_{n-1}}+\sum_{s_{n}=1}^{h}Q\left( s_{n-1},s_{n}\right)
L_{n,n-1}^{\left( \mathcal{H}_{n-1}\right) \prime }N_{n\mid n-1}^{\left( 
\mathcal{H}_{n}\right) }L_{n,n-1}^{\left( \mathcal{H}_{n-1}\right) }.
\end{eqnarray*}%
Here%
\begin{equation*}
L_{n,n-1}^{\left( \mathcal{H}_{n-1}\right) }=T_{s_{n}}\left( I-K_{n-1\mid
n-2}^{\left( \mathcal{H}_{n-1}\right) }Z_{s_{n-1}}\right) .
\end{equation*}

Continuing in the same way, we can summarise the procedures in the following
algorithm.

\begin{algorithm}
State Smoothing
\end{algorithm}

\textbf{Step 0.} Initialise the smoother by setting $r_{n\mid n-1}^{\left( 
\mathcal{H}_{n}\right) }=Z_{s_{n}}^{\prime }\left[ F_{n\mid n-1}^{\left( 
\mathcal{H}_{n}\right) }\right] ^{-1}v_{n\mid n-1}^{\left( \mathcal{H}%
_{n}\right) }$, $\ \ N_{n\mid n-1}^{\left( \mathcal{H}_{n}\right)
}=Z_{s_{n}}^{\prime }\left[ F_{n\mid n-1}^{\left( \mathcal{H}_{n}\right) }%
\right] ^{-1}Z_{s_{n}},$ $\mathcal{H}_{n}\in \mathbb{H}_{N,n}$. Here $\alpha
_{n\mid n}^{\left( \mathcal{H}_{n}\right) }$ is the output of the
corresponding filter at $t=n$.

\textbf{Step 1.} Compute the smoothed estimates of the state vector an the
MSE for each history $\mathcal{H}_{t}$, using recursion:%
\begin{eqnarray*}
L_{t+1,t}^{\left( \mathcal{H}_{t}\right) } &=&T_{s_{t+1}}\left( I-K_{t\mid
t-1}^{\left( \mathcal{H}_{t}\right) }Z_{s_{t}}\right) \\
r_{t\mid t-1}^{\left( \mathcal{H}_{t}\right) } &=&Z_{s_{t}}^{\prime }\left[
F_{t\mid t-1}^{\left( \mathcal{H}_{t}\right) }\right] ^{-1}v_{t\mid
t-1}^{\left( \mathcal{H}_{t}\right) }+\sum_{s_{t+1}=1}^{h}Q\left(
s_{t},s_{t+1}\right) L_{t+1,t}^{\left( \mathcal{H}_{t}\right) \prime
}r_{t+1\mid t}^{\left( \mathcal{H}_{t+1}\right) }, \\
N_{t\mid t-1}^{\left( \mathcal{H}_{t}\right) } &=&Z_{s_{t}}^{\prime }\left[
F_{t\mid t-1}^{\left( \mathcal{H}_{t}\right) }\right] ^{-1}Z_{s_{t}}+%
\sum_{s_{t+1}=1}^{h}Q\left( s_{t},s_{t+1}\right) L_{t+1,t}^{\left( \mathcal{H%
}_{t}\right) \prime }N_{t+1|t}^{\left( \mathcal{H}_{t+1}\right)
}L_{t,t+1}^{\left( \mathcal{H}_{t}\right) } \\
\alpha _{t\mid n}^{\left( \mathcal{H}_{t}\right) } &=&\alpha _{t\mid
t-1}^{\left( \mathcal{H}_{t}\right) }+P_{t\mid t-1}^{\left( \mathcal{H}%
_{t}\right) }r_{t\mid t-1}^{\left( \mathcal{H}_{t}\right) }, \\
P_{t\mid n}^{\left( \mathcal{H}_{t}\right) } &=&P_{t\mid t-1}^{\left( 
\mathcal{H}_{t}\right) }-P_{t\mid t-1}^{\left( \mathcal{H}_{t}\right)
}N_{t\mid t-1}^{\left( \mathcal{H}_{t}\right) }P_{t\mid t-1}^{\left( 
\mathcal{H}_{t}\right) },
\end{eqnarray*}%
for $t=n-1,n-2,...,1.$

Use the smoothed probabilities, $\left \{ \mu _{t\mid n}^{(\mathcal{H}%
_{t})}\right \} $, to compute the smoothed state vectors and MSE matrices:%
\begin{eqnarray*}
x_{t\mid n} &=&\sum_{\mathcal{H}_{t}}\mu _{t\mid n}^{(\mathcal{H}%
_{t})}\alpha _{t\mid n}^{\left( \mathcal{H}_{t}\right) }, \\
P_{t\mid n} &=&\sum_{\mathcal{H}_{t}}\mu _{t\mid n}^{(\mathcal{H}%
_{t})}P_{t\mid n}^{\left( \mathcal{H}_{t}\right) }.
\end{eqnarray*}

\section{Validating the Filters\label{Sec Simulations}}

\subsection{Model and Parameterisation}

To compare the performance of filters and smoothers, we use the model
developed in\newline
\citet*{FVGQRR}, hereafter referred to as FGR2015. It is a relatively
standard medium-scale New Keynesian DSGE model, which we modify to
investigate the aspects of good luck and good policy.

The model consists of a household sector, firms, and a monetary authority.
Households derive utility from consumption relative to their habit stock and
from leisure. They supply differentiated labour to monopolistically
competitive firms and choose wages subject to Calvo wage-setting friction.
Firms produce differentiated output using capital, labour, and a neutral
technology process. They set prices, also subject to Calvo pricing
frictions. The capital stock evolves in the usual way, except for the
inclusion of embodied technology in new investment goods. The model is
closed by imposing a Taylor-type rule for the monetary authority. We present
the full specification of the model in Appendix \ref{App model}.

We base the structural parameters of the model on the estimates reported in
FGR2015; see column (1) in Table C1 in Appendix \ref{App model}. Our
treatment of policy and shock volatilities is different from FGR2015, who
estimated a single-regime nonlinear policy function and a single-regime
stochastic volatility process. We introduce two Markov-switching processes
into the model. The first, $S_{P,t}$, governs policy parameters in the
following monetary policy rule: 
\begin{equation}
\frac{r_{t}}{r_{ss}}=\left( \frac{r_{t-1}}{r_{ss}}\right) ^{\gamma
_{r}\left( S_{P,t}\right) }\left( \left( \frac{\pi _{t}}{\pi _{\text{targ}}}%
\right) ^{\gamma _{\pi }\left( S_{P,t}\right) }\left( \frac{Y_{d,t}}{\lambda
_{yd}Y_{d,t-1}}\right) ^{\gamma _{y}\left( S_{P,t}\right) }\right)
^{1-\gamma _{r}\left( S_{P,t}\right) }\exp \left( \sigma _{\xi }\left(
S_{V,t}\right) \varepsilon _{\xi ,t}\right) .  \label{policyrule}
\end{equation}%
The literature typically categorizes monetary policy approaches into hawkish
and dovish modes, characterized by more and less aggressive responses to
inflation, respectively. Accordingly, we assume that the $\gamma -$%
parameters are high in state $S_{P,t}=1$ (hawkish state) and low in state $%
S_{P,t}=2$ (dovish state). We explain below how we chose these values. The
second two-state process, $S_{V,t},$ governs the shock volatilities for all
shocks, including the policy shock in equation (\ref{policyrule}).

\subsection{Monte-Carlo Simulations Design}

In our simulations, we aim to differentiate between periods of infrequent
large shocks and periods of more frequent regular shocks. We set the
probability of remaining in the low volatility state to 0.95. This
parameterisation implies an average of 20 quarters between high shocks, with
a standard deviation of 19 quarters.\footnote{%
If probability to leave one of the two Markov states is $q$, then the
expected length of stay in this state is $1/q$ with the standard deviation
of $\sqrt{1-q}/q.$} This probability accurately reflects the fact that
recessions in the US have occurred approximately every 8-10 years since the
end of World War II. We set the probability of staying in the high
volatility state to 0.8, resulting in an average duration of high shock
periods of 5 quarters (with a standard deviation of 4.5 quarters).
Interpreting periods of large shocks as recessions suggests that a typical
recession lasts slightly for less than a year, a duration that our
parameterisation appropriately captures.

In formulating our policy model, we applied considerations similar to those
used in the assumptions in the shock volatility experiments. The existing
literature tends to report that hawkish policies have been predominant since
the 1980s, spanning approximately 40 years.\footnote{%
See e.g. \citet{BianchiMelosiAER2017}, \citet*{CKL2017}.} However,
considering the data starting from 1955 and acknowledging the evident dovish
tendencies since 2008, we infer that the time split between these regimes is
roughly equal. Therefore, we assume symmetric diagonal elements in the
transition probability matrix. As the benchmark case, we calibrate the
probability to remain in either of these states at 0.95. This implies an
average of 20 quarters between policy changes, allowing for a wide range of
durations between policy shifts. In addition, we consider an alternative
calibration, with this probability set to 0.1. All transition matrices are
presented in Table \ref{Tab regimes}.%
%TCIMACRO{\TeXButton{B}{\begin{table}[h] \centering}}%
%BeginExpansion
\begin{table}[h] \centering%
%EndExpansion
\caption{Parameterisation of shock and policy regimes}%
\begin{tabular}{ccc}
\multicolumn{3}{c}{Transition matrices} \\ \hline\hline
Shocks & Benchmark Case & Alternative Case \\ \hline
$P_{s}=\overset{}{\left[ 
\begin{array}{cc}
0.95 & 0.05 \\ 
0.2 & 0.8%
\end{array}%
\right] }$ & $P_{p}^{I}=\underset{}{\left[ 
\begin{array}{cc}
0.95 & 0.05 \\ 
0.05 & 0.95%
\end{array}%
\right] }$ & $P_{p}^{II}=\left[ 
\begin{array}{cc}
0.9 & 0.1 \\ 
0.1 & 0.9%
\end{array}%
\right] $ \\ \hline\hline
&  &  \\ 
\multicolumn{3}{c}{Parameters of Taylor Rule:} \\ \hline\hline
& Hawkish Feedback & Dovish Feedback \\ \hline
$\gamma _{\pi }^{Base}$ & $1.7$ & $0.9$ \\ 
$\gamma _{\pi }^{Altern.}$ & $1.5$ & $0.9$ \\ \hline\hline
\end{tabular}%
\label{Tab regimes}%
%TCIMACRO{\TeXButton{E}{\end{table}}}%
%BeginExpansion
\end{table}%
%EndExpansion

As for the policy coefficients that are time-varying (or depend on the
state), we describe the hawkish policy mode with feedback on inflation $%
\gamma _{\pi }^{Base}=1.7$ in the hawkish state and $\gamma _{\pi
}^{Base}=0.9$ in the dovish state, consistent with findings in other studies%
\footnote{%
See, e.g. \citet{BianchiAER2012}, \citet*{Chang2021}, \citet*{CLL2022}.}. We
also consider an alternative parameterisation where these two feedbacks are
less distinct, as shown in Table \ref{Tab regimes}. In these simulations, we
keep the feedback on output and the interest rate smoothing parameter the
same in both hawkish and dovish states.

As reported in column (1) in Table C1, the standard deviations of all shocks
in the low-volatility state, $S_{V,t}=1,$ are set to be equal to the mean
estimates of corresponding variables in FGR2015, and they are doubled in the
high-volatility state, $S_{V,t}=2.$

In order to generate artificial data, we solve and simulate this non-linear
model using a perturbation approach with the functional iteration algorithm
developed for RISE$^{\textcopyright}$ \ (\citealp{MaihRISE2015}).

We chose to generate 500 samples of 1,000 observations each. We consider
output growth, price inflation, wage inflation, the Federal Funds rate, and
the relative price of investment goods as observable variables. The latent
variables are listed in Table \ref{Tab MRSE1K} and other relevant tables. We
then use the simulation results to investigate the performance of the
discussed filters, controlling for the sample length. Within each sample, we
use the initial 300 observations as a proxy for a typical real-life scenario
with post-WW2 quarterly data, where the influence of initial conditions can
be substantial. Additionally, we analyze the full sample of 1000
observations, in which we expect the impact of initial conditions to be
significantly diminished.

\subsection{Results}

\subsubsection{Evaluation Criteria}

We need some criteria to rank the filters for practical purpose, based on
their accuracy and speed. For accuracy, or goodness-of-fit, in our exercise
we cannot use measures linked to the likelihood $L_{t}=\log f\left(
y_{t}\mid Y_{t-1}\right) $ returned by the filters. This is because
different filters employ different approximations when computing the
likelihood, and so comparison based on this measure is not compelling for
comparison of the filters. An alternative and, perhaps, more straightforward
approach in our case is to use root mean squared errors (RMSE) for each
latent variable $\alpha _{t}$, given by the formula%
\begin{equation*}
\mathcal{R}_{\varphi }=\frac{1}{n_{sim}}\sum_{i=1}^{n_{sim}}\sqrt{\frac{1}{n}%
\sum_{t=1}^{n}\left( \frac{\alpha _{t}-\alpha _{\varphi }}{\alpha _{ss}}%
\right) ^{2}}.
\end{equation*}%
We\ present the comparison of the accuracy of the filters based on the
updated variables ($\varphi =t\mid t$) and smoothed variables ($\varphi
=t\mid n)$ in Tables \ref{Tab MRSE1K}-\ref{Tab Missp}.\footnote{%
In computing RMSEs, we normalise all variables, except state probabilities,
by their steady-state levels, as in this model the steady state is identical
for all regimes.} Here, $n$ is the length of each data sample, and $n_{sim}$
is the number of simulations.

\subsubsection{The Best Performing Filter}

Table \ref{Tab MRSE1K} shows the results for four filters: the IMM(1) and
GPB(N) for $N=1,2,3$, which includes the KN filter as it is equivalent to
GPB(2).\footnote{%
Appendix \ref{App Filters} presents selected filtering and smoothing
algorithms in a form convenient for implementation.} Our simulations reveal
that increasing the order of the GPB(N) filter beyond N=3 offers no
practical value. We do not present results for IMM(2) as it does not
noticeably improve accuracy of the IMM(1).

%TCIMACRO{\TeXButton{B}{\begin{table}[h] \centering}}%
%BeginExpansion
\begin{table}[h] \centering%
%EndExpansion
\caption{MRSEs for update variables from four filters}%
\begin{tabular}{lcccc|cccc}
\hline\hline
& \multicolumn{4}{c|}{{\small Absolute RMSEs }$R_{t\mid t}$} & 
\multicolumn{4}{|c}{\small Relative RMSEs} \\ \cline{2-9}
& {\small IMM(1)} & {\small GPB(2)} & {\small GPB(1)} & {\small GPB(3)} & 
{\small IMM(1)} & {\small GPB(2)} & {\small GPB(1)} & {\small GPB(3)} \\ 
&  & {\small KN} &  &  &  & {\small KN} &  &  \\ 
{\small variables} & {\small (1)} & {\small (2)} & {\small (3)} & {\small (4)%
} & {\small (5)} & {\small (6)} & {\small (7)} & {\small (8)} \\ \hline
{\small consumption} & {\small 0.032} & {\small 0.032} & {\small 0.033} & 
{\small 0.032} & {\small 1.0003} & {\small 1} & {\small 1.030} & {\small %
1.0002} \\ 
{\small capital} & {\small 0.226} & {\small 0.226} & {\small 0.238} & 
{\small 0.226} & {\small 1.0001} & {\small 1} & {\small 1.051} & {\small %
1.0006} \\ 
{\small output} & {\small 0.030} & {\small 0.030} & {\small 0.030} & {\small %
0.030} & {\small 1.0005} & {\small 1.00004} & {\small 1.016} & {\small 1} \\ 
{\small real wage} & {\small 0.002} & {\small 0.002} & {\small 0.002} & 
{\small 0.002} & {\small 1.0002} & {\small 1} & {\small 1.035} & {\small %
1.0003} \\ 
{\small Tobin's Q} & {\small 0.010} & {\small 0.010} & {\small 0.010} & 
{\small 0.010} & {\small 1.0001} & {\small 1} & {\small 1.025} & {\small %
1.0001} \\ 
{\small investment} & {\small 0.302} & {\small 0.302} & {\small 0.321} & 
{\small 0.302} & {\small 1} & {\small 1.00004} & {\small 1.065} & {\small %
1.0010} \\ 
{\small lab supply} & {\small 0.029} & {\small 0.029} & {\small 0.030} & 
{\small 0.029} & {\small 1.0005} & {\small 1.00004} & {\small 1.016} & 
{\small 1} \\ 
{\small pref shock} & {\small 0.043} & {\small 0.043} & {\small 0.044} & 
{\small 0.043} & {\small 1.0001} & {\small 1} & {\small 1.015} & {\small %
1.00006} \\ 
{\small lab sup shock} & {\small 0.070} & {\small 0.070} & {\small 0.072} & 
{\small 0.070} & {\small 1.0005} & {\small 1.0001} & {\small 1.019} & 
{\small 1} \\ 
{\small tech shock} & {\small 0.001} & {\small 0.001} & {\small 0.001} & 
{\small 0.001} & {\small 1.0004} & {\small 1} & {\small 1.066} & {\small %
1.0015} \\ \hline
{\small shock reg. probs} & {\small 0.265} & {\small 0.265} & {\small 0.266}
& {\small 0.265} & {\small 1.00002} & {\small 1} & {\small 1.002} & {\small %
1.0004} \\ 
{\small policy reg. probs} & {\small 0.335} & {\small 0.335} & {\small 0.344}
& {\small 0.335} & {\small 1} & {\small 1.00001} & {\small 1.023} & {\small %
1.00003} \\ \hline\hline
\end{tabular}%
\label{Tab MRSE1K}%
%TCIMACRO{\TeXButton{E}{\end{table}}}%
%BeginExpansion
\end{table}%
%EndExpansion

We focus on the \textit{updated} variables, as these variables contribute to
the likelihood used in estimation. The average RMSEs (denoted as $\mathcal{R}%
_{t\mid t}$) for all 500 draws in the Monte Carlo experiment are presented
in columns (1)-(4) of Table \ref{Tab MRSE1K}. They vary in magnitude,
reflecting findings similar to those in \citet{BinningMaih2015}, where it is
observed that highly persistent latent variables, such as capital, pose
greater challenges for reconstruction.

The relative RMSEs in columns (5)-(8) are computed by dividing the RMSE for
each variable by the lowest RMSE for that particular variable across
investigated filters. In other words, for the best-performing filter it has
the value of 1, while for all other filters its value is greater than one.

\begin{figure} [!h]
    \begin{center}
    \includegraphics[width = 0.95\textwidth]{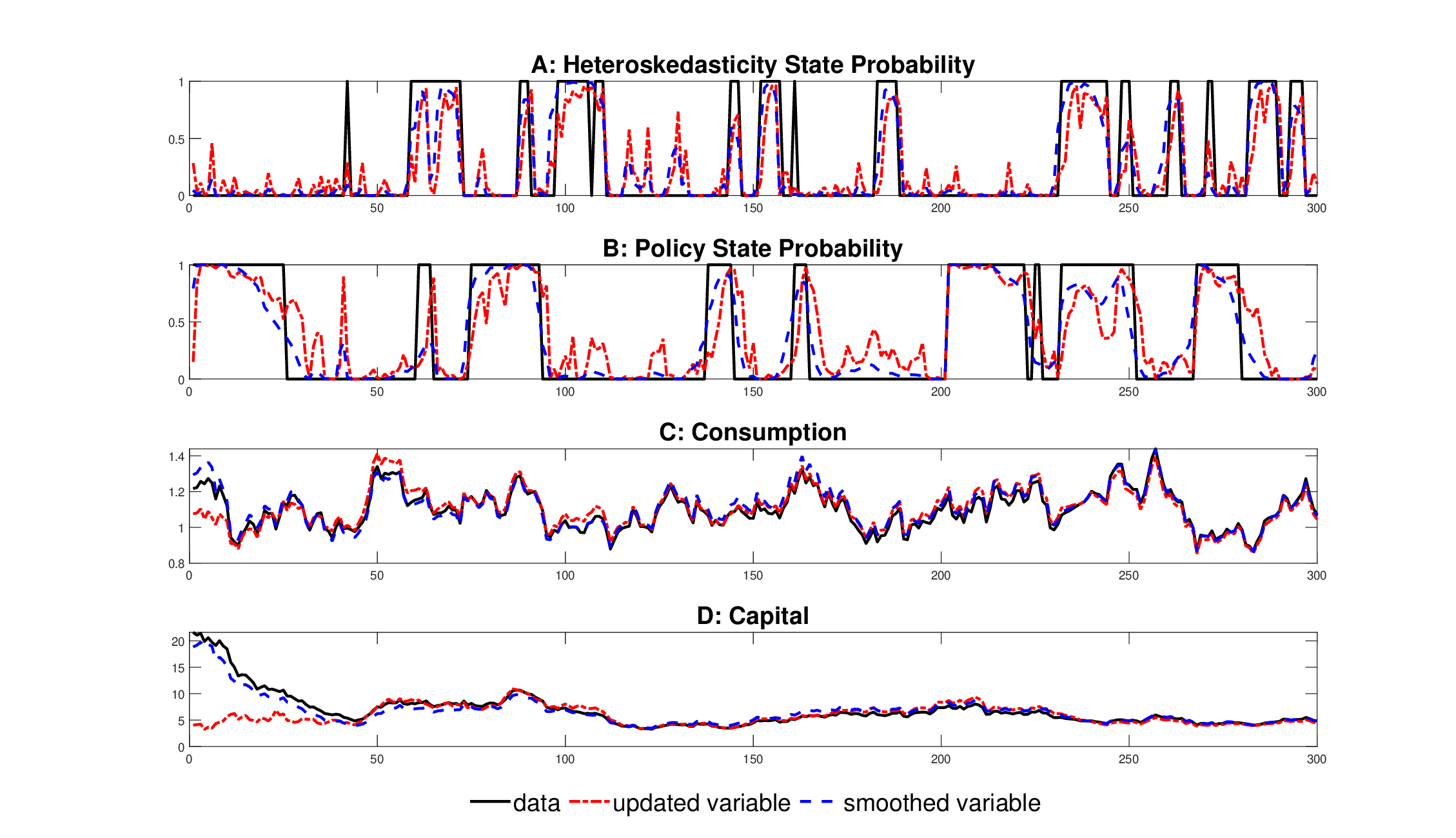}
    \caption{Updating and smoothing produced by IMM filter.}
    \label{fig:FigUandS}
    \end{center}
\end{figure}

%TCIMACRO{\TeXButton{B}{\begin{table}[h] \centering}}%
%BeginExpansion
\begin{table}[h] \centering%
%EndExpansion
\caption{Computational times for filtering 1000 observations}%
\begin{tabular}{lcc|cccc}
\hline\hline
\multicolumn{3}{c}{A: IMM{\small (1)}\ vs. GPB(2)-KN} & \multicolumn{4}{|c}{
B: Relative speed} \\ \hline
& updating & updating and &  & \multicolumn{3}{c}{updating} \\ 
& only & smoothing &  & \multicolumn{3}{c}{only} \\ 
& sec & sec &  & ratio &  & ratio \\ \hline
&  &  & GPB(1) & 0.28 & \multicolumn{1}{|c}{} &  \\ 
IMM{\small (1)} & 0.27 & 1.49 & GPB(2) & 1 & \multicolumn{1}{|c}{IMM(1)} & 1
\\ 
\multicolumn{1}{c}{GPB(2)} & 1.38 & 2.59 & GPB(3) & 4.21 & 
\multicolumn{1}{|c}{IMM(2)} & 5.81 \\ 
&  &  & GPB(4) & 17.74 & \multicolumn{1}{|c}{IMM(3)} & 52.70 \\ 
&  &  & GPB(5) & 79.97 & \multicolumn{1}{|c}{IMM(4)} & 691.14 \\ \hline\hline
\end{tabular}%
\label{Tab Speed}%
%TCIMACRO{\TeXButton{E}{\end{table}}}%
%BeginExpansion
\end{table}%
%EndExpansion

These RMSEs\ are visualised by a red dash-dotted line in Figure \ref{Fig
UandS}, illustrating the recovery of latent variables and probabilities of
being in a particular Markov state in one particular simulation. The true
values are shown by the solid black lines. We plot only the initial 300
observations for clarity of visualisation.

In terms of computational burden, IMM{\small (1)} has an advantage over
KN-GPB(2) as it works with $h$ histories, rather than with $h^{2}$. Panel A
of Table \ref{Tab Speed} shows indicative computational times for these two
filters, both independently and in conjunction with the corresponding
smoother.\footnote{%
These numbers are achieved on a Ryzen 3950X with 64GB RAM using MATLAB$^{%
\textcopyright}$ R2022b.} These times serve merely as an indication of
computational speed, as all filters implemented in RISE$^{\textcopyright}$
perform additional tasks beyond algorithm computation, which, even if not
used, hold the potential to reduce speed.\footnote{%
This includes checking for and accomodating properties such as time-varying
states and parameters, missing observations, nonstationarity, and
occasionally-binding constraints. See Appendix \ref{App RISE} for further
notes on implementation.} It is important to note that at the estimation
stage, where speed is particularly crucial, smoothing is not applied.

One can see in panel B of Table \ref{Tab Speed} how quickly the computation
burden of IMM rises with higher orders. This is because the algorithm keeps
track of all possible histories of fixed length where the histories contain
all possible combinations of regimes in every time period of recursion.
Therefore, the number of terms containing probabilities of different
combinations of regimes is much bigger in IMM than in GPB.

One can reduce the computational burden in IMM(N) by using approximation
similar to (\ref{Markov}). However, because of the repeated use the
disrepancy created by this approximation accumulates and leads to lower
accuracy, especially for higher $N$.

On the balance between accuracy and speed, it is clear that IMM{\small (1)}
dominates the GPB(N) family of filters. Based on that, we argue that in
practical applications, IMM{\small (1)} is the best. In what follows, we use
IMM to denote IMM(1).

\subsubsection{Updating and Smoothing}

To visualize the effect of smoothing, in Figure \ref{FigUandS} we plot
updated and smoothed probabilities and latent variables, alongside their
true values. When comparing all lines, we see that smoothing reduces
high-frequency noise in the recovered probabilities and often helps to
identify the timing of regime changes more accurately. It is also apparent
that the initial gap between the updated and actual values of latent
variables is substantially corrected, although the improvement is not
uniform across the entire sample. There are time periods where smoothing
does not improve these particular variables at all. Notably, while the
impact of initial conditions on the filter's effectiveness for economic
variables is very clear, such an effect is less prominent for probabilities.%
%TCIMACRO{\TeXButton{B}{\begin{table}[h] \centering}}%
%BeginExpansion
\begin{table}[h] \centering%
%EndExpansion
\caption{Accuracy improvement by smoothing, $1-(R_{t|T}/R_{t|t})$}%
\begin{tabular}{lcccc}
\hline\hline
vars: & IMM & GPB(2) & GPB(1) & GPB(3) \\ \hline
consumption & 0.29 & 0.27 & 0.31 & 0.27 \\ 
capital & 0.21 & 0.21 & 0.23 & 0.21 \\ 
output & 0.42 & 0.37 & 0.47 & 0.38 \\ 
real wage & 0.18 & 0.18 & 0.18 & 0.18 \\ 
Tobin's Q & 0.25 & 0.25 & 0.23 & 0.25 \\ 
investment & 0.30 & 0.28 & 0.32 & 0.28 \\ 
labour supply & 0.42 & 0.37 & 0.47 & 0.38 \\ 
preference shock & 0.14 & 0.13 & 0.14 & 0.13 \\ 
labour supply shock & 0.37 & 0.33 & 0.40 & 0.34 \\ 
technology shock & 0.10 & 0.10 & 0.04 & 0.10 \\ \hline
shock state probs & 0.15 & 0.15 & 0.15 & 0.15 \\ 
policy state probs & 0.18 & 0.18 & 0.16 & 0.18 \\ \hline\hline
\end{tabular}%
\label{Tab Smooth}%
%TCIMACRO{\TeXButton{E}{\end{table}}}%
%BeginExpansion
\end{table}%
%EndExpansion

Figure \ref{FigUandS} illustrates the work of smoothing for the IMM filter.
For all considered filters, Table \ref{Tab Smooth} reports the fractions of
RMSEs that is removed by smoothing, $1-\frac{\mathcal{R}_{t\mid T}}{\mathcal{%
R}_{t\mid t}}.$

While the RMSEs for some variables are only improved by 4\%, for some others
the improvement is as large as 47\%, and the average improvement is about
25\%. For state probabilities, the average improvement is 16\%. Notably, the
less computationally intensive smoothers for IMM and GPB(1) improve accuracy
better than more complex algorithms for higher order of GPB filters.

\subsubsection{Information}

The following experiments aim to address the importance of working with
longer series of data. In the first experiment, we use the IMM filter; other
filters show very similar results. Column (1) of Panel A in Table \ref{Tab
Info}\ presents RMSEs for updated variables using the first 300 observations
in each simulation. Columns (2) and (3) present RMSEs for smoothed variables
for the first 300 observations, where smoothing starts from the last of
these 300 observations in column (2) and from the last of all 1000
observations in \ column (3).\ The comparison is consistent with the
intuition that more information improves accuracy. However, this intuition
is not necessarily true in our setting because filtering and smoothing
procedures involve numerous approximations. It is remarkable that, despite
these approximations, the improvement in accuracy is substantial: close to
10\% for some variables.

%TCIMACRO{\TeXButton{B}{\begin{table}[h] \centering}}%
%BeginExpansion
\begin{table}[h] \centering%
%EndExpansion
\caption{Importance of Information. Panel A: RMSE for updated (1) and
smoothed (2,3) variables. Panel B: Accuracy improvement by smoothing.}%
\begin{tabular}{lccc|cccc}
\hline\hline
& \multicolumn{3}{c}{Panel A} & \multicolumn{4}{|c}{Panel B} \\ 
sample size & 300 & 300 & 300 & 300 & 250 & 200 & 100 \\ 
& $\mathcal{R}_{t\mid t}$ & $\mathcal{R}_{t\mid 300}$ & $\mathcal{R}_{t\mid
1000}$ & 1-$\frac{\mathcal{R}_{t\mid 300}}{\mathcal{R}_{t\mid t}}$ & 1-$%
\frac{\mathcal{R}_{t\mid 250}}{\mathcal{R}_{t\mid t}}$ & 1-$\frac{\mathcal{R}%
_{t\mid 200}}{\mathcal{R}_{t\mid t}}$ & 1-$\frac{\mathcal{R}_{t\mid 100}}{%
\mathcal{R}_{t\mid t}}$ \\ 
vars: & (1) & (2) & (3) & (1) & (2) & (3) & (4) \\ \hline
consumption & 0.051 & 0.036 & 0.033 & 0.30 & 0.30 & 0.30 & 0.26 \\ 
capital & 0.352 & 0.277 & 0.266 & 0.21 & 0.21 & 0.20 & 0.14 \\ 
output & 0.050 & 0.029 & 0.026 & 0.43 & 0.43 & 0.42 & 0.33 \\ 
real wage & 0.003 & 0.002 & 0.002 & 0.17 & 0.17 & 0.17 & 0.11 \\ 
Tobin's Q & 0.010 & 0.008 & 0.008 & 0.26 & 0.26 & 0.26 & 0.22 \\ 
investment & 0.466 & 0.317 & 0.300 & 0.32 & 0.32 & 0.32 & 0.27 \\ 
labour supply & 0.049 & 0.028 & 0.026 & 0.43 & 0.42 & 0.42 & 0.32 \\ 
pref. shock & 0.051 & 0.042 & 0.041 & 0.18 & 0.19 & 0.20 & 0.19 \\ 
lab. supp. shock & 0.110 & 0.066 & 0.061 & 0.40 & 0.40 & 0.40 & 0.30 \\ 
techn. shock & 0.002 & 0.001 & 0.001 & 0.10 & 0.10 & 0.11 & 0.11 \\ \hline
shock state probs & 0.265 & 0.225 & 0.225 & 0.15 & 0.15 & 0.15 & 0.15 \\ 
policy state probs & 0.335 & 0.275 & 0.275 & 0.18 & 0.18 & 0.18 & 0.19 \\ 
\hline\hline
\end{tabular}%
\label{Tab Info}%
%TCIMACRO{\TeXButton{E}{\end{table}}}%
%BeginExpansion
\end{table}%
%EndExpansion

Panel B of Table \ref{Tab Info} shows an improvement in the RMSEs from
smoothing \ obtained in the second experiment. Here we explore different
sample sizes, $n\in \{300,250,200,100\}$. We know that RMSEs of updated
variables are larger in shorter samples. One might expect that the efficacy
of smoothing will also deteriorate in shorter samples. However, Panel B
reveals that this is not the case, except for $n=100$. This suggests that
the sample size should be in excess of 100 and, perhaps, at least 200 to
ensure that the smoother improves accuracy.

Another important observation is that there is no sample size effect for the
smoothing of state probabilities, as can be seen in the last two rows of
Panel B. This is consistent with the results presented in Figure \ref{Fig
UandS}, where the effect of initial conditions is only observed for latent
economic variables: in a shorter sample initial conditions play bigger role,
but this not the case for probabilities.

\subsubsection{Policy States}

For the next experiment, we simulate artificial data for less and more
distinct policy states as measured by different probabilities of remaining
in a given policy state in the next period, and by larger difference in the
feedback coefficient $\gamma _{\pi }$ in two policy states. Table \ref{Tab
MRSEpolicy} reports the results.

%TCIMACRO{\TeXButton{B}{\begin{table}[h] \centering}}%
%BeginExpansion
\begin{table}[h] \centering%
%EndExpansion
\caption{RMSE for Pobability of Hawkish Policy State.}%
\begin{tabular}{ccccc}
\hline\hline
policy description & $\gamma _{\pi }$ & $P_{HH}=P_{DD}$ & updated & smoothed
\\ \hline
more distinct states, less distinct feedback & 1.5 & 0.95 & 0.366 & 0.310 \\ 
less distinct states, less distinct feedback & 1.5 & 0.9 & 0.412 & 0.380 \\ 
more distinct states, more distinct feedback & 1.7 & 0.95 & 0.335 & 0.275 \\ 
less distinct states, more distinct feedback & 1.7 & 0.9 & 0.386 & 0.349 \\ 
\hline\hline
\end{tabular}%
\label{Tab MRSEpolicy}%
%TCIMACRO{\TeXButton{E}{\end{table}}}%
%BeginExpansion
\end{table}%
%EndExpansion

We conclude that the greater the difference between the states, the better
is their identification. This is true for both updated and smoothed policy
state probabilities.

\subsubsection{Model Misspecification}

We investigate several cases of the model misspecification relevant for the
Markov-switching nature of our model. We assume that the true
data-generating process contains two Markov-switching processes as described
above, but a researcher only considers one of them, either in policy or in
volatility.

\begin{figure} [!h]
    \begin{center}
    \includegraphics[width = 0.95\textwidth]{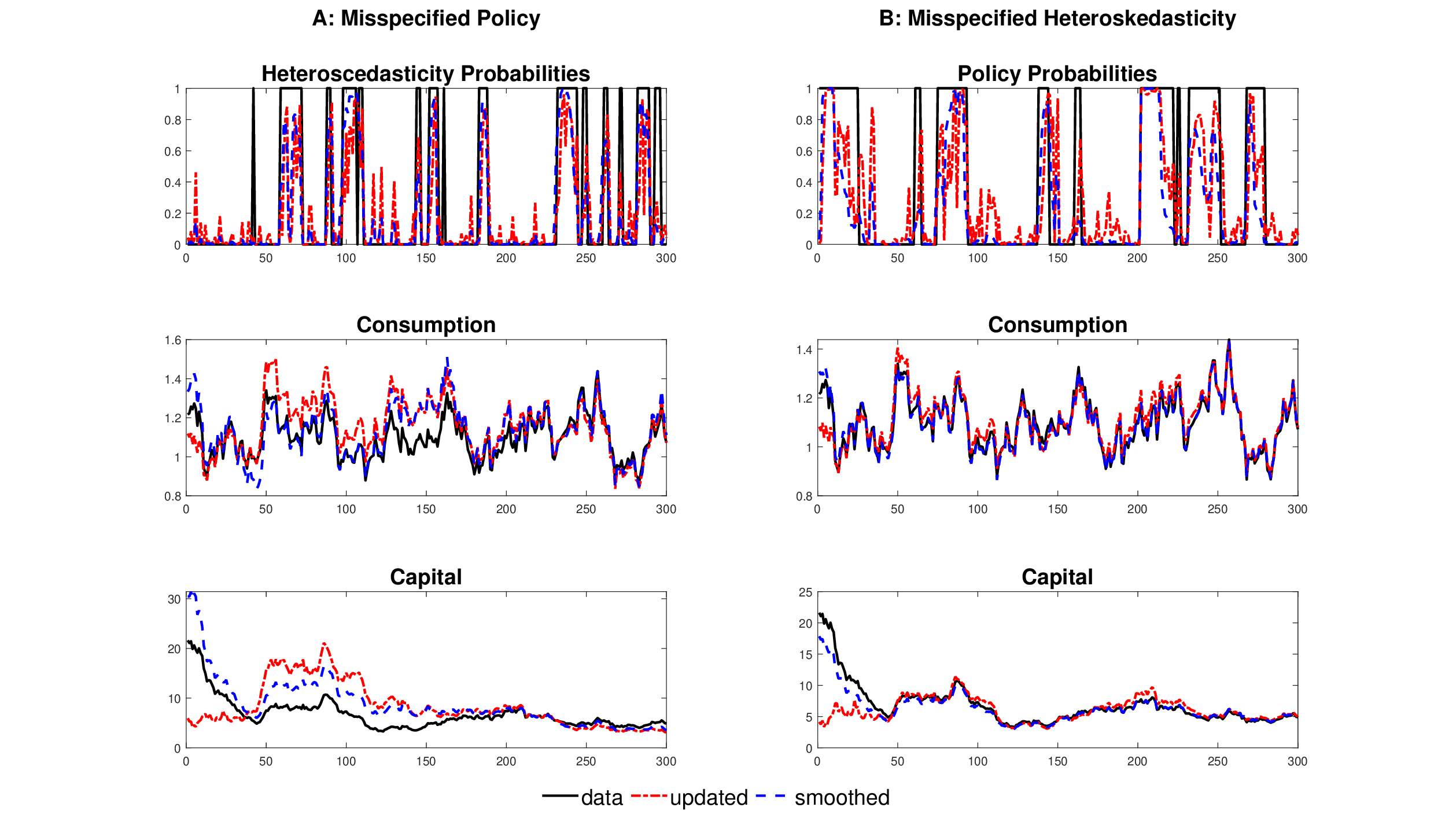}
    \caption{Model Misspecifications.}
    \label{fig:FigMissp}
    \end{center}
\end{figure}

%TCIMACRO{\TeXButton{B}{\begin{table}[h] \centering}}%
%BeginExpansion
\begin{table}[h] \centering%
%EndExpansion
\caption{MRSEs of smoothed variables $R_{t|1000}$. }%
\begin{tabular}{lcccc}
\hline\hline
vars: & No & Missp-d & Missp-d & Missp-d \\ 
& missp-b & Policy H & Shocks L & Shocks H \\ 
& (1) & (2) & (3) & (4) \\ \hline
consumption & $\underset{[0.032]}{0.023}$ & $\underset{[0.070]}{0.057}$ & $%
\underset{[0.036]}{0.022}$ & $\underset{[0.035]}{0.021}$ \\ 
capital & $\underset{[0.226]}{0.178}$ & $\underset{[1.683]}{1.472}$ & $%
\underset{[0.259]}{0.186}$ & $\underset{[0.239]}{0.164}$ \\ 
output & $\underset{[0.030]}{0.017}$ & $\underset{[0.038]}{0.020}$ & $%
\underset{[0.032]}{0.012}$ & $\underset{[0.031]}{0.011}$ \\ 
real wage & $\underset{[0.002]}{0.002}$ & $\underset{[0.005]}{0.005}$ & $%
\underset{[0.002]}{0.002}$ & $\underset{[0.02]}{0.002}$ \\ 
Tobin's Q & $\underset{[0.010]}{0.007}$ & $\underset{[0.037]}{0.036}$ & $%
\underset{[0.013]}{0.011}$ & $\underset{[0.010]}{0.007}$ \\ 
investment & $\underset{[0.302]}{0.210}$ & $\underset{[2.533]}{2.151}$ & $%
\underset{[0.349]}{0.224}$ & $\underset{[0.317]}{0.182}$ \\ 
labour supply & $\underset{[0.029]}{0.017}$ & $\underset{[0.036]}{0.019}$ & $%
\underset{[0.031]}{0.011}$ & $\underset{[0.031]}{0.011}$ \\ 
preference shock & $\underset{[0.043]}{0.037}$ & $\underset{[0.113]}{0.112}$
& $\underset{[0.051]}{0.045}$ & $\underset{[0.044]}{0.037}$ \\ 
labour supply shock & $\underset{[0.070]}{0.044}$ & $\underset{[0.199]}{0.174%
}$ & $\underset{[0.080]}{0.045}$ & $\underset{[0.073]}{0.035}$ \\ 
technology shock & $\underset{[0.001]}{0.001}$ & $\underset{[0.002]}{0.002}$
& $\underset{[0.001]}{0.001}$ & $\underset{[0.001]}{0.001}$ \\ \hline
shock state probs & $\underset{[0.265]}{0.224}$ & $\underset{[0.280]}{0.247}$
& -- & -- \\ 
policy state probs & $\underset{[0.335]}{0.275}$ & -- & $\underset{[0.408]}{%
0.390}$ & $\underset{[0.339]}{0.282}$ \\ \hline\hline
\end{tabular}

Note: MRSEs of updated variables $\mathcal{R}_{t\mid t}$ are in square
brackets.\label{Tab Missp}%
%TCIMACRO{\TeXButton{E}{\end{table}}}%
%BeginExpansion
\end{table}%
%EndExpansion
In the first scenario, we assume that the researcher believes in a single
policy stance: the hawkish state.\footnote{%
This is the common assumption in constant-parameter DSGE model estimations,
see e.g. \citet*{CKL2017}} We then execute the filtering and smoothing
algorithm with $\gamma _{\pi }^{Base}$, which is the correct policy feedback
but only for one of the two policy Markov states. Panel A in Figure \ref{Fig
Missp} illustrates the outcomes of updating and smoothing for the same
simulation, focusing on the initial 300 observations. The updated and
smoothed variables are shown to accurately identify the patterns in the
data, although with larger errors than in the correctly specified model.

One can get further insights from Table \ref{Tab Missp} by comparing the
RMSEs for updated and smoothed variables obtained from 500 simulations of
the entire sample of 1000 observations in columns (1) and (2). Despite
effectively handling data patterns, RMSEs for smoothed latent variables show
a substantial increase. However, the RMSEs for smoothed heteroskedasticity
state probability show only a small increase, as Figure \ref{FigMissp} also
demonstrates.

In the second scenario, we assume that the researcher believes the
volatility is always low. Panel B of Figure \ref{FigMissp}, which shows the
results of this scenario, confirms that the filter correctly identifies the
patterns in the data. Column (3) in Table \ref{Tab Missp} further supports
this observation, as it shows a much smaller increase in RMSEs compared to
column (1), and even smaller estimation errors for some variables. At the
same time, the RMSE for smoothed probability of a policy state is higher
than in the correctly specified model. This suggests that incorrectly
specifying shock volatilities significantly worsens the identification of
policy states.

In the final experiment, reported in column (4), we revisit the second
scenario, but this time we assume that the researcher believes the
volatility is always high. Although the RMSEs for updated latent variables
in column (4) are higher than those in column (1), the RMSEs for smoothed
variables are sometimes lower than in the correctly specified model. The
unexpectedly superior performance of the misspecified model after smoothing
can be attributed to the larger variance of shocks. By allowing for a large
variance in the shocks distribution, it accommodates both large and small
shocks. The smoother then revises the estimated values using the complete
sample and adjusts the estimates by factoring in information about the
realized shocks.

\subsection{Interim Summary}

Overall, the findings in this section demonstrate the effectiveness of the
canonical IMM filter, particularly when combined with the appropriate
smoother, in enhancing the accuracy and efficiency of Bayesian estimation of
state-space models.

The canonical IMM outperforms the Kim and Nelson filter in terms of
computational speed while delivering comparable accuracy. The implementation
of the new smoothing algorithm with the IMM filter substantially enhances
precision in estimating latent variables, reducing errors by approximately
25\%{}{}. We do not find any substantial improvement in accuracy when\ using
higher order filters in our example. It is hard to predict whether the same
will be true for other models.

Our simulations confirm that, despite approximations, adding more
information improves the performance of the suggested filtering-smoothing
procedure. We find that, as long as the sample length remains above 200
observations, there is no reduction in the smoother's efficacy in reducing
RMSEs for updated variables. We find that the filter identifies
probabilities of more distinct policy regimes with higher accuracy.

Finally, we demonstrate that we can still successfully recover latent
variables even when the policy or shock volatility regimes in the model are
misspecified.

Having established the superiority of the canonical IMM paired with the
matching smoother, we focus on this filter and smoother in the empirical
application.

\section{Empirical Application\label{Sec Empirics}}

In this section, we further investigate the practicality of the IMM filter
with the corresponding smoother. We estimate a modified version of the
FGR2015 model but using the same data for 1959Q2-2013Q4 as in that paper
(see Table C2 in Appendix \ref{App Params}). In our estimation we impose
relatively wide priors and use the \textit{Artificial Bee Colony} algorithm
by \citet{KarabogaBasturk2007} for global optimisation.

Table \ref{Tab regimes est} displays the estimated mode of the distribution
of transition probabilities and policy parameters. We present the remaining
parameters in Table C1 in Appendix \ref{App Params}.

%TCIMACRO{\TeXButton{B}{\begin{table}[h] \centering}}%
%BeginExpansion
\begin{table}[h] \centering%
%EndExpansion
\caption{Estimation of parameters that govern the two Markov processes}%
\begin{tabular}{cc}
\multicolumn{2}{c}{Transition matrices} \\ \hline\hline
Shocks & Policy \\ \hline
$P_{s}=\overset{}{\left[ 
\begin{array}{cc}
0.939\,05 & 0.060946 \\ 
0.04625 & 0.95375%
\end{array}%
\right] }$ & $P_{p}=\underset{}{\left[ 
\begin{array}{cc}
0.983\,96 & 0.016039 \\ 
0.043428 & 0.956\,57%
\end{array}%
\right] }$ \\ \hline\hline
&  \\ 
\multicolumn{2}{c}{Parameters of Taylor Rule:} \\ \hline\hline
Hawkish Feedback $\gamma _{\pi }$ & 1.6574 \\ 
Dovish Feedback $\gamma _{\pi }$ & 0.93984 \\ \hline\hline
\end{tabular}%
\label{Tab regimes est}%
%TCIMACRO{\TeXButton{E}{\end{table}}}%
%BeginExpansion
\end{table}%
%EndExpansion

We note that both regime-switching processes are highly persistent, and
therefore their identification is likely to be correct, as suggested by our
simulation results. The policy process, in particular, shows that there is
only a 2\% probability of leaving the hawkish state.

\begin{figure} [!h]
    \begin{center}
    \includegraphics[width = 0.95\textwidth]{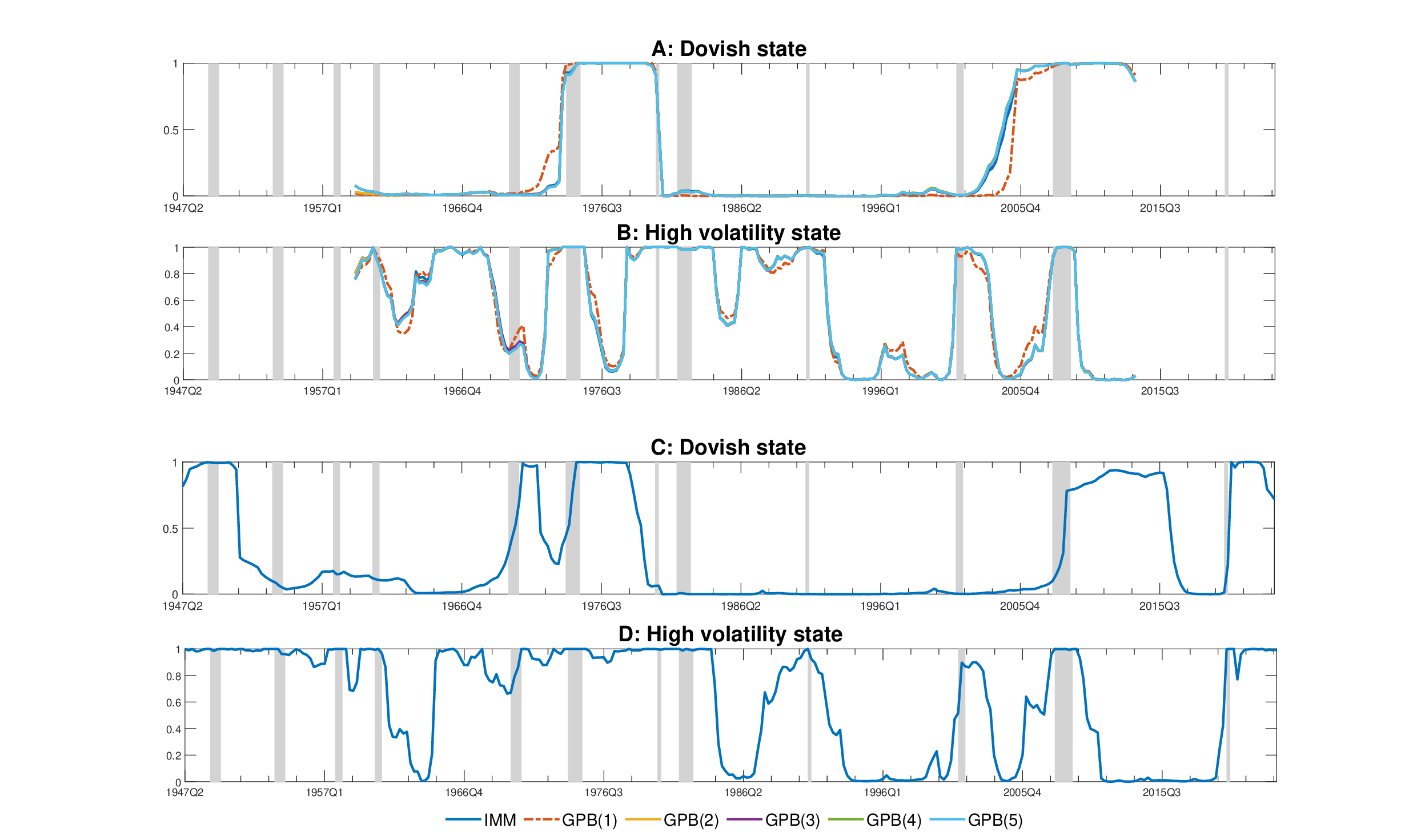}
    \caption{Smoothed State Probabilities.}
    \label{fig:FigEmpirics}
    \end{center}
\end{figure}

Panels A and B in Figure \ref{FigEmpirics} report the smoothed
probabilities of being in the dovish state and the high volatility state. We
used the canonical IMM at the estimation stage and six different filters at
the filtering stage.

One can notice that the lines plotted for six filters are very close to one
another. The GPB filters of order 2 to 5 produce nearly identical results
and they are also extremely close to those produced by the canonical IMM.
This suggests, first, that using a more computationally intensive
higher-order GPB filter does not necessarily improve regime identification
compared to the KN-GPB(2) filter, and, second, that the canonical IMM and
the KN filter are practically identical in accuracy. While GPB(1) stands out
as less accurate, it still identifies all main events similarly to the other
filters.

Panel A shows the probability of being in the dovish policy state. Note that
it indicates that our approach succeded in identifying all major changes in
the US post-war policy stance: the Great Inflation, the Volcker
Disinflation, the Great Moderation, the Great Financial Crisis, and the
subsequent Zero Lower Bound (ZLB) period. We did not assume a special regime
for ZLB monetary policy but identify this period as a dovish state.

Panel B shows the probability of being in the high volatility state. Our
approach correctly identifies most of the recessions and suggests that the
pre-1990s period experienced larger shocks than the more recent past.

For panels C and D the dataset includes the period 1947Q2-2023Q3 (see
Appendix \ref{App Params} for details). The extended data covers a longer
period, adding observations at the beginning, which should improve the
identification of the Great Inflation episode, and at the end, which
includes the post-Covid period with rising inflation in 2022-23. In these
two panels we only show the results obtained using the IMM filter (with the
associated smoother) to this extended dataset.

The message is similar to what is suggested by panels A and B. We identify
the dovish state during the ZLB and a shift to hawkish policy a year after
the ZLB lift-off. In addition, we see the return to the dovish policy during
the Covid-19 pandemic which lasted until 2023Q1, at which time tough
measures against inflation were taken. The post-Covid period is also
characterised by relatively large shocks.

\section{Conclusions\label{Sec Conclusions}}

Our focus in this paper has been on improving multiple-regime Bayesian
filtering techniques, alongside the development of multiple-regime smoothers.

We introduced the family of IMM filters, along with an extension of the Kim
and Nelson filter, to accommodate tracking of longer regime histories. In
addition, we developed a robust smoothing algorithm that can be adapted to
these extended filters. Our simulation exercises demonstrate that the IMM
filter with our proposed smoother deliver the best combination of
computational speed and accuracy in a prototypical macroeconomic application
of Bayesian filtering.

Our paper provides a comprehensive toolkit for researchers working with
complex macroeconomic models. We demonstrate its practical relevance in an
empirical application using a NK DSGE\ model with long U.S. macroeconomic
time series.

\bibliographystyle{chicago}
\bibliography{demo}

\pagebreak

\appendix

\begin{center}
{\LARGE Online Appendix}

{\Large to}

{\Large On Bayesian Filtering for Markov Regime Switching Models}

{\large by}

{\large Nigar Hashimzade Oleg Kirsanov Tatiana Kirsanova Junior Maih}{\Large %
\ }
\end{center}

\section{Selected Algorithms\label{App Filters}}

Let $M_{j,t}^{y}=\left \{ Z_{j,t},c_{y,j,t},T_{j,t},c_{\alpha
,j,t},g_{j,t},R_{j,t};y_{t}\right \} $ be state-space system matrices for
regime $j$ and information at time $t.$ Let $\mathcal{K}()$ be a KF
operator. The filtering algorithms are summarised in Tables A1-A2. Smoothing
algorithms are summarised in Table A3.

\begin{center}
\begin{tabular}{c|c}
\multicolumn{2}{c}{Table A1$\underset{}{}$: GPB Filtering Algorithms} \\ 
\hline\hline
GPB(1) & GPB(2) \\ \hline
\multicolumn{2}{c}{Regime probabilities} \\ \hline
$\mu _{t\mid t}^{j}:=\func{Pr}\left[ s_{t}=j\mid Y_{t}\right] $ & $\mu
_{t-1\mid t}^{ij}=\func{Pr}\left[ s_{t-1}=i,s_{t}=j\mid Y_{t}\right] $ \\ 
\multicolumn{1}{l|}{} & \multicolumn{1}{|l}{$\mu _{t\mid
t}^{j}=\sum_{i=1}^{h}\mu _{t-1\mid t}^{ij}$} \\ \hline
\multicolumn{2}{c}{Initialisation} \\ \hline
\multicolumn{1}{l|}{$\alpha _{t-1\mid t-1},P_{t-1\mid t-1},\mu _{t-1\mid
t-1}^{j}$} & \multicolumn{1}{|l}{$\alpha _{t-1\mid t-1}^{i},P_{t-1\mid
t-1}^{i},\mu _{t-1\mid t-1}^{i}$} \\ \hline
\multicolumn{2}{c}{Filtering and Updating} \\ \hline
\multicolumn{1}{l|}{$\left[ \alpha _{t\mid t-1}^{j},P_{t\mid
t-1}^{j},v_{t\mid t-1}^{j},\alpha _{t\mid t}^{j},P_{t\mid t}^{j}\right] $} & 
\multicolumn{1}{|l}{$\left[ \alpha _{t\mid t-1}^{ij},P_{t\mid
t-1}^{ij},v_{t\mid t-1}^{ij},\alpha _{t\mid t}^{ij},P_{t\mid t}^{ij}\right] $%
} \\ 
$=\mathcal{K}\left( M_{j,t}^{y};\alpha _{t-1\mid t-1},P_{t-1\mid t-1}\right) 
$ & $=\mathcal{K}\left( M_{j,t}^{y};\alpha _{t-1\mid t-1}^{i},P_{t-1\mid
t-1}^{i}\right) $ \\ 
\multicolumn{1}{l|}{$\Lambda _{t}^{j}=(2\pi )^{-t/2}\left \vert
F_{j,t}\right \vert ^{-1/2}e^{-\frac{1}{2}v_{t\mid t-1}^{j\prime
}F_{j.t}^{^{-1}}v_{t\mid t-1}^{j}}$} & \multicolumn{1}{|l}{$\Lambda
_{t}^{ij}=(2\pi )^{-t/2}\left \vert F_{i,j,t}\right \vert ^{-1/2}e^{-\frac{1%
}{2}v_{t\mid t-1}^{ij\prime }F_{i,j,t}^{^{-1}}v_{t\mid t-1}^{ij}}$} \\ \hline
\multicolumn{2}{c}{Collapsing (dimension reduction) and Probabilities update}
\\ \hline
\multicolumn{1}{l|}{$\mu _{t\mid t}^{j}=\frac{\Lambda
_{t}^{j}\sum_{i=1}^{h}Q_{t-1,t}^{ij}\mu _{t-1\mid t-1}^{i}}{%
\sum_{k,m=1}^{h}\Lambda _{t}^{m}Q_{t-1,t}^{km}\mu _{t-1\mid t-1}^{k}}$} & 
\multicolumn{1}{|l}{$\mu _{t-1\mid t}^{ij}=\frac{\Lambda
_{t}^{ij}Q_{t-1,t}^{ij}\mu _{t-1\mid t-1}^{i}}{\sum_{k=1}^{h}\Lambda
_{t}^{kj}Q_{t-1,t}^{kj}\mu _{t-1\mid t-1}^{k}}$} \\ 
\multicolumn{1}{l|}{$\alpha _{t\mid t}=\sum_{j=1}^{h}\mu _{t\mid
t}^{j}\alpha _{t\mid t}^{j}$} & \multicolumn{1}{|l}{$\alpha _{t\mid
t}^{j}=\sum_{i=1}^{h}\mu _{t-1\mid t}^{ij}\alpha _{t\mid t}^{ij}$} \\ 
\multicolumn{1}{l|}{$P_{t\mid t}=\sum_{i=1}^{h}\mu _{t\mid t}^{j}\left(
P_{t\mid t}^{j}\right. $} & \multicolumn{1}{|l}{$P_{t\mid
t}^{j}=\sum_{i=1}^{h}\mu _{t-1\mid t}^{ij}\left( P_{t\mid t}^{ij}\right. $}
\\ 
$\left. +\left( \alpha _{t\mid t}-\alpha _{t\mid t}^{j}\right) \left( \alpha
_{t\mid t}-\alpha _{t\mid t}^{j}\right) ^{\prime }\right) $ & $\left.
+\left( \alpha _{t\mid t}^{j}-\alpha _{t\mid t}^{ij}\right) \left( \alpha
_{t\mid t}^{j}-\alpha _{t\mid t}^{ij}\right) ^{\prime }\right) $ \\ 
\multicolumn{1}{l|}{} & \multicolumn{1}{|l}{$\mu _{t\mid
t}^{j}=\sum_{i=1}^{h}\mu _{t-1\mid t}^{ij}$} \\ \hline\hline
\end{tabular}

\begin{tabular}{c|c}
\multicolumn{2}{c}{Table A2:$\underset{}{}$ IMM Filtering Algorithms} \\ 
\hline\hline
IMM(1) & IMM(2) \\ \hline
\multicolumn{2}{c}{Regime probabilities} \\ \hline
$\mu _{t\mid t}^{j}:=\func{Pr}\left[ s_{t}=j\mid Y_{t}\right] $ & $\mu
_{t-1\mid t}^{ij}=\func{Pr}\left[ s_{t-1}=i,s_{t}=j\mid Y_{t}\right] $ \\ 
\hline
\multicolumn{2}{c}{Initialisation} \\ \hline
\multicolumn{1}{l|}{$\alpha _{t-1\mid t-1}^{i},P_{t-1\mid t-1}^{i},\mu
_{t-1\mid t-1}^{i}$} & \multicolumn{1}{|l}{$\alpha _{t-1\mid
t-1}^{ki},P_{t-1\mid t-1}^{ki},\mu _{t-1\mid t-1}^{ki}$} \\ \hline
\multicolumn{2}{c}{Mixing (dimension reduction)} \\ \hline
\multicolumn{1}{l|}{$\mu _{t-1\mid t-1}^{i\mid j}=\frac{Q_{t-1,t}^{ij}\mu
_{t-1\mid t-1}^{i}}{\sum_{k=1}^{h}Q_{t-1,t}^{kj}\mu _{t-1\mid t-1}^{k}}$} & 
\multicolumn{1}{|l}{$\mu _{t-1\mid t-1}^{ki\mid ij}=\frac{\mathcal{Q}%
_{k(h-1)+i,i(h-1)+j}\mu _{t-2\mid t-1}^{ki}}{\sum_{m,l=1}^{h}\mathcal{Q}%
_{m(h-1)+l,l(h-1)+j}\mu _{t-2\mid t-1}^{ml}}$} \\ 
\multicolumn{1}{l|}{$\alpha _{t-1\mid t-1}^{0j}=\sum_{i=1}^{h}\mu _{t-1\mid
t-1}^{i\mid j}\alpha _{t-1\mid t-1}^{i}$} & \multicolumn{1}{|l}{$\alpha
_{t-1\mid t-1}^{0ij}=\sum_{k,i=1}^{h}\mu _{t-1\mid t-1}^{ki\mid ij}\alpha
_{t-1\mid t-1}^{ki}$} \\ 
\multicolumn{1}{l|}{$P_{t-1\mid t-1}^{0j}=\sum_{i=1}^{h}\mu _{t-1\mid
t-1}^{i\mid j}\left( P_{t-1\mid t-1}^{i}\right. $} & \multicolumn{1}{|l}{$%
P_{t-1\mid t-1}^{0ij}=\sum_{k,i=1}^{h}\mu _{t-1\mid t-1}^{ki\mid ij}\left(
P_{t-1\mid t-1}^{ki}\right. $} \\ 
$+\left( \alpha _{t-1\mid t-1}^{i}-\alpha _{t-1\mid t-1}^{0i}\right) $ & $%
+\left( \alpha _{t-1\mid t-1}^{ki}-\alpha _{t-1\mid t-1}^{0ki}\right) $ \\ 
$\left. \times \left( \alpha _{t-1\mid t-1}^{i}-\alpha _{t-1\mid
t-1}^{0i}\right) ^{\prime }\right) $ & $\left. \times \left( \alpha
_{t-1\mid t-1}^{ki}-\alpha _{t-1\mid t-1}^{0ki}\right) ^{\prime }\right) $
\\ \hline
\multicolumn{2}{c}{Filtering and Updating} \\ \hline
\multicolumn{1}{l|}{$\left[ \alpha _{t\mid t-1}^{j},P_{t\mid
t-1}^{j},v_{t\mid t-1}^{j},\alpha _{t\mid t}^{j},P_{t\mid t}^{j}\right] $} & 
\multicolumn{1}{|l}{$\left[ \alpha _{t\mid t-1}^{ij},P_{t\mid
t-1}^{ij},v_{t\mid t-1}^{ij},\alpha _{t\mid t}^{ij},P_{t\mid t}^{ij}\right] $%
} \\ 
\multicolumn{1}{l|}{$=\mathcal{K}\left( M_{j,t}^{y};\alpha _{t-1\mid
t-1}^{0j},P_{t-1\mid t-1}^{0j}\right) $} & \multicolumn{1}{|l}{$=\mathcal{K}%
\left( M_{j,t}^{y};\alpha _{t-1\mid t-1}^{0ij},P_{t-1\mid t-1}^{0ij}\right) $%
} \\ 
\multicolumn{1}{l|}{$\Lambda _{t}^{j}=(2\pi )^{-t/2}\left \vert
F_{j,t}\right \vert ^{-1/2}e^{-\frac{1}{2}v_{t\mid t-1}^{j\prime
}F_{j,t}^{^{-1}}v_{t\mid t-1}^{j}}$} & \multicolumn{1}{|l}{$\Lambda
_{t}^{ij}=(2\pi )^{-t/2}\left \vert F_{ij,t}\right \vert ^{-1/2}e^{-\frac{1}{%
2}v_{t\mid t-1}^{ij\prime }F_{ij,t}^{^{-1}}v_{t\mid t-1}^{ij}}$} \\ \hline
\multicolumn{2}{c}{Probabilities update} \\ \hline
\multicolumn{1}{l|}{$\mu _{t\mid t}^{j}=\frac{\Lambda
_{t}^{j}\sum_{i=1}^{h}Q_{t-1,t}^{ij}\mu _{t-1\mid t-1}^{i}}{%
\sum_{k,m=1}^{h}\Lambda _{t}^{m}Q_{t-1,t}^{km}\mu _{t-1\mid t-1}^{k}}$} & 
\multicolumn{1}{|l}{$\mu _{t-1\mid t}^{ij}=\frac{\Lambda
_{t}^{ij}Q_{t-1,t}^{ij}\mu _{t-1\mid t-1}^{i}}{\sum_{k=1}^{h}\Lambda
_{t}^{kj}Q_{t-1,t}^{kj}\mu _{t-1\mid t-1}^{k}}$} \\ \hline\hline
\end{tabular}

\begin{tabular}{cc}
\multicolumn{2}{c}{Table A3: $\underset{}{}$Smoothing Algorithms} \\ 
\hline\hline
GPB(1) and IMM(1) & GPB(2) and IMM(2) \\ \hline
\multicolumn{2}{c}{\textit{Smoothed Probabilities}} \\ 
\multicolumn{2}{c}{$\mu _{t\mid n}^{j}=\sum_{k=1}^{h}\mu _{t+1\mid n}^{k}%
\frac{\mu _{t\mid t}^{j}Q_{t,t+1}^{jk}}{\sum_{m=1}^{h}Q_{t,t+1}^{mk}\mu
_{t\mid t}^{m}}$} \\ \hline
\multicolumn{2}{c}{\textit{Smoothed States}} \\ 
\multicolumn{2}{c}{Initialisaion} \\ \hline
\multicolumn{1}{l}{$r_{n\mid n-1}^{i}=Z_{i,n}^{\prime }F_{i,n}^{-1}v_{n\mid
n-1}^{i}$} & \multicolumn{1}{|l}{$r_{n\mid n-1}^{ij}=Z_{j,n}^{\prime
}F_{ij,n}^{-1}v_{n\mid n-1}^{ij}$} \\ \hline
\multicolumn{2}{c}{Recursion} \\ \hline
\multicolumn{1}{l}{$L_{t+1,t}^{ij}=T_{j,t+1}\left( I-K_{i,t}Z_{i,t}\right) $}
& \multicolumn{1}{|l}{$L_{t+1,t}^{ijk}=T_{k,t+1}\left(
I-K_{i,j,t}Z_{j,t}\right) $} \\ 
\multicolumn{1}{l}{$r_{t\mid t-1}^{i}=Z_{i,t}^{\prime }F_{i,t}^{-1}v_{t\mid
t-1}^{i}$} & \multicolumn{1}{|l}{$r_{t\mid t-1}^{ij}=Z_{t,j}^{\prime
}F_{i,j,t}^{-1}v_{t\mid t-1}^{ij}$} \\ 
\multicolumn{1}{l}{$+\sum_{j=1}^{h}Q_{t,t+1}^{ij}L_{t+1,t}^{ij\prime
}r_{t+1\mid t}^{j}$} & \multicolumn{1}{|l}{$%
+\sum_{k=1}^{h}Q_{t,t+1}^{jk}L_{t+1,t}^{ijk\prime }r_{t+1\mid t}^{jk}$} \\ 
$\alpha _{t\mid n}^{i}=\alpha _{t\mid t-1}^{i}+P_{t\mid t-1}^{i}r_{t\mid
t-1}^{i}$ & \multicolumn{1}{|c}{$\alpha _{t\mid n}^{ij}=\alpha _{t\mid
t-1}^{ij}+P_{t\mid t-1}^{ij}r_{t\mid t-1}^{ij}$} \\ \hline
\multicolumn{2}{c}{Merge states} \\ \hline
\multicolumn{1}{l}{} & \multicolumn{1}{|l}{$\alpha _{t\mid
n}^{j}=\sum_{i=1}^{h}\mu _{t-1|t}^{ij}\alpha _{t\mid n}^{ij}$} \\ 
\multicolumn{1}{l}{$\alpha _{t\mid n}=\sum_{i=1}^{h}\mu _{t\mid n}^{i}\alpha
_{t\mid n}^{i}$} & \multicolumn{1}{|l}{$\alpha _{t\mid n}=\sum_{i=1}^{h}\mu
_{t\mid n}^{i}\alpha _{t\mid n}^{i}$} \\ \hline\hline
\end{tabular}
\end{center}

\section{The Model\label{App model}}

This section summarises the model in \citet{FVGQRR}. We present the list of
variables and all the model equations. We then present parameterisation of
the model used in Section \ref{Sec Simulations}, and estimated parameters
obtained in the empirical investigation discussed in Section \ref{Sec
Empirics}.

\begin{center}
\begin{tabular}{llll}
\multicolumn{4}{c}{Table B1: List of Variables} \\ \hline\hline
$d_{t}$ & Shifter to intertemp. preference & $C_{t}$ & Consumption \\ 
$G_{t}$ & Government consumption & $\Lambda _{t}$ & Marginal utility of
consumption \\ 
$r_{t}$ & gross nominal interest rate & $R_{kt}$ & Rental rate of capital \\ 
$\pi _{t}$ & Gross inflation & $\phi _{t}$ & Cost of use of capital \\ 
$Q_{t}$ & Tobin's Q & $\phi _{t}^{\prime }$ & derivative of the capital adj.
cost \\ 
$X_{t}$ & Investment & $u_{t}$ & capital utilization \\ 
$s_{t}$ & Investment adjustment cost & $s_{t}^{\prime }$ & derivative of
invest. adj. cost \\ 
$f_{t}$ & Calvo wage parameter & $W_{\ast ,t}$ & Optimal real wage \\ 
$W_{t}$ & real wage & $l_{d,t}$ & labor demand \\ 
$\varphi _{t}$ & labor supply shifter & $\pi _{\ast w,t}$ & Relative optimal
real wage \\ 
$g_{1,t}$ & Calvo price process 1 & $\pi _{\ast ,t}$ & Relative Price \\ 
$g_{2,t}$ & Calvo price process 2 & $mc_{t}$ & Real marginal cost \\ 
$Y_{d,t}$ & Output & $v_{p,t}$ & Price dispersion \\ 
$K_{t}$ & Capital & $A_{t}$ & Neutral technology \\ 
$Z_{t}$ & Combined technology & $MU_{t}$ & Investment-specific tech. level
\\ 
$v_{w,t}$ & Wage dispersion & $l_{t}$ & hours worked/labor supply \\ 
$\varepsilon _{\xi ,t}$ & Monetary policy shock, scale $\sigma _{\xi }$ & $%
\varepsilon _{\varphi ,t}$ & labor supply shock, with scale $\sigma
_{\varphi }$ \\ 
$\varepsilon _{g,t}$ & Government spending shock, scale $\sigma _{g}$ & $%
\varepsilon _{\mu ,t}$ & Invest.-spec. technology shock, scale $\sigma _{\mu
}$ \\ 
$\varepsilon _{d,t}$ & Preference shock, scale $\sigma _{d}$ & $\varepsilon
_{A,t}$ & Neutral technology shock, scale $\sigma _{a}$ \\ \hline\hline
\end{tabular}

\begin{tabular}{ll}
\multicolumn{2}{c}{Table B2$\underset{}{}$: Model Equations} \\ \hline\hline
\multicolumn{2}{l}{\textit{Households}} \\ \hline
Capital accum-n & $K_{t}=\left( 1-\delta \right) K_{t-1}+MU_{t}\left( 1-s%
\left[ \frac{X_{t}}{X_{t-1}}\right] \right) X_{t}$ \\ 
FOC consum-n & $\frac{d_{t}}{C_{t}-hC_{t-1}}-h\beta E_{t}\frac{d_{t+1}}{%
C_{t+1}-hC_{t}}=\Lambda _{t}$ \\ 
FOC bonds & $\Lambda _{t}=\beta E_{t}\Lambda _{t+1}\frac{r_{t}}{\pi _{t+1}}$
\\ 
FOC capital util. & $R_{kt}=\frac{\phi ^{\prime }\left[ u_{t}\right] }{MU_{t}%
}$ \\ 
FOC capital & $Q_{t}=\beta E_{t}\frac{\Lambda _{t+1}}{\Lambda _{t}}\left(
\left( 1-\delta \right) Q_{t+1}+R_{kt+1}u_{t+1}-\frac{\phi \left[ u_{t+1}%
\right] }{MU_{t+1}}\right) $ \\ 
Capital util-n & $\phi \left[ u\right] =\phi _{1}\left( u-1\right) +\frac{%
\phi _{2}}{2}\left( u-1\right) ^{2}$ \\ 
\ \ its derivative & $\phi ^{\prime }\left[ u\right] =\phi _{1}+\frac{\phi
_{2}}{2}\left( u-1\right) $ \\ 
FOC investment & $1=Q_{t}MU_{t}\left( 1-s\left[ \frac{X_{t}}{X_{t-1}}\right]
-s^{\prime }\left[ \frac{X_{t}}{X_{t-1}}\right] \frac{X_{t}}{X_{t-1}}\right) 
$ \\ 
& $+\beta E_{t}Q_{t+1}MU_{t+1}\frac{\Lambda _{t+1}}{\Lambda _{t}}s^{\prime }%
\left[ \frac{X_{t+1}}{X_{t}}\right] \left( \frac{X_{t+1}}{X_{t}}\right) ^{2}$
\\ 
Invest. adj. cost & $s\left[ \frac{X_{t}}{X_{t-1}}\right] =\frac{\kappa }{2}%
\left( \frac{X_{t}}{X_{t-1}}-\lambda _{x}\right) ^{2}$ \\ 
\ \ its derivative & $s^{\prime }\left[ \frac{X_{t}}{X_{t-1}}\right] =\kappa
\left( \frac{X_{t}}{X_{t-1}}-\lambda _{x}\right) $ \\ \hline
\multicolumn{2}{l}{\textit{Firms}} \\ \hline
Wage helper 1 & $f_{t}=\frac{\eta -1}{\eta }\left( W_{\ast ,t}\right)
^{1-\eta }\Lambda _{t}W_{t}^{\eta }l_{d,t}+\beta \theta _{w}E_{t}\left( 
\frac{\pi _{t}^{\chi _{w}}}{\pi _{t+1}}\right) ^{1-\eta }\left( \frac{%
W_{\ast ,t+1}}{W_{\ast ,t}}\right) ^{\eta -1}f_{t+1}$ \\ 
Wage helper 2 & $f_{t}=\psi d_{t}\varphi _{t}\pi _{\ast w,t}^{-\eta \left(
1+\vartheta \right) }l_{d,t}^{\left( 1+\vartheta \right) }+\beta \theta
_{w}E_{t}\left( \frac{\pi _{t}^{\chi _{w}}}{\pi _{t+1}}\right) ^{-\eta
\left( 1+\vartheta \right) }\left( \frac{W_{\ast ,t+1}}{W_{\ast ,t}}\right)
^{\eta \left( 1+\vartheta \right) }f_{t+1}$ \\ 
Wage setting & $\pi _{\ast w,t}=\frac{W_{\ast ,t}}{W_{t}}$ \\ 
Wage dynamics & $1=\theta _{w}\left( \frac{\pi _{t-1}^{\chi _{w}}}{\pi _{t}}%
\right) ^{1-\eta }\left( \frac{W_{t-1}}{W_{t}}\right) ^{1-\eta }+\left(
1-\theta _{w}\right) \pi _{\ast w,t}^{1-\eta }$ \\ 
Wage dispersion & $v_{w,t}=\theta _{w}\left( \frac{W_{t-1}}{W_{t}}\frac{\pi
_{t-1}^{\chi _{w}}}{\pi _{t}}\right) ^{-\eta }v_{w,t-1}+\left( 1-\theta
_{w}\right) \pi _{\ast w,t}^{-\eta }$ \\ 
Price helper 1 & $g_{1,t}=\Lambda _{t}mc_{t}Y_{d,t}+\beta \theta
_{p}E_{t}\left( \frac{\pi _{t}^{\chi }}{\pi _{t+1}}\right) ^{-\varepsilon
}g_{1,t+1}$ \\ 
Price helper 2 & $g_{2,t}=\Lambda _{t}\pi _{\ast ,t}Y_{d,t}+\beta \theta
_{p}E_{t}\left( \frac{\pi _{t}^{\chi }}{\pi _{t+1}}\right) ^{1-\varepsilon
}\left( \frac{\pi _{\ast ,t}}{\pi _{\ast ,t+1}}\right) g_{2,t+1}$ \\ 
Price setting & $\varepsilon g_{1,t}=\left( \varepsilon -1\right) g_{2,t}$
\\ 
Price dynamics & $1=\theta _{p}\left( \frac{\pi _{t-1}^{\chi }}{\pi _{t}}%
\right) ^{1-\varepsilon }+\left( 1-\theta _{p}\right) \pi _{\ast
,t}^{1-\varepsilon }$ \\ 
Price dispersion & $v_{p,t}=\theta _{p}\left( \frac{\pi _{t-1}^{\chi }}{\pi
_{t}}\right) ^{-\varepsilon }v_{p,t-1}+\left( 1-\theta _{p}\right) \pi
_{\ast ,t}^{-\varepsilon }$ \\ \hline\hline
\multicolumn{2}{r}{continued on the next page}%
\end{tabular}
\end{center}

\begin{tabular}{ll}
\multicolumn{2}{c}{Table B2$\underset{}{}$: Model Equations -- continued} \\ 
\hline\hline
\multicolumn{2}{l}{\textit{Market Clearing and Policy}} \\ \hline
Production function & $Y_{d,t}=\frac{A_{t}\left( u_{t}K_{t-1}\right)
^{\alpha }\left( l_{d,t}\right) ^{1-\alpha }-\phi _{y}Z_{t}}{v_{p,t}}$ \\ 
Capital-labor ratio & $\frac{u_{t}K_{t-1}}{l_{d,t}}=\frac{\alpha }{1-\alpha }%
\frac{W_{t}}{R_{kt}}$ \\ 
Aggregate labour & $l_{t}=v_{w,t}l_{d,t}$ \\ 
Resource constraint & $Y_{d,t}=C_{t}+G_{t}+X_{t}+\frac{\phi \left[ u_{t}%
\right] }{MU_{t}}K_{t-1}$ \\ 
Marginal costs & $mc_{t}=\left( \frac{1}{1-\alpha }\right) ^{1-\alpha
}\left( \frac{1}{\alpha }\right) ^{\alpha }\frac{W_{t}^{1-\alpha
}R_{kt}^{\alpha }}{A_{t}}$ \\ 
Taylor rule & $\frac{r_{t}}{r_{ss}}=\left( \frac{r_{t-1}}{r_{ss}}\right)
^{\gamma _{r}}\left( \left( \frac{\pi _{t}}{\pi _{\text{targ}}}\right)
^{\gamma _{\pi }}\left( \frac{Y_{d,t}}{\lambda _{yd}Y_{d,t-1}}\right)
^{\gamma _{y}}\right) ^{1-\gamma _{r}}\exp \left( \sigma _{\xi }\varepsilon
_{\xi ,t}\right) $ \\ 
Government spending & $\log \left( \frac{G_{t}}{Z_{t}}\right) =\left( 1-\rho
_{g}\right) \log g+\rho _{g}\log \left( \frac{G_{t-1}}{Z_{t-1}}\right)
+\sigma _{g}\varepsilon _{g,t}$ \\ \hline
\multicolumn{2}{l}{\textit{Exogenous processes}} \\ \hline
Intertemporal preference & $\log \left( d_{t}\right) =\rho _{d}\log \left(
d_{t-1}\right) +\sigma _{d}\varepsilon _{d,t}$ \\ 
Labor supply & $\log \left( \varphi _{t}\right) =\rho _{\varphi }\log \left(
\varphi _{t-1}\right) +\sigma _{\varphi }\varepsilon _{\varphi ,t}$ \\ 
Investment-spec. technology & $MU_{t}=MU_{t-1}\exp \left( \lambda _{\mu
}+\sigma _{\mu }\varepsilon _{\mu ,t}\right) $ \\ 
Neutral technology & $A_{t}=A_{t-1}\exp \left( \lambda _{a}+\sigma
_{a}\varepsilon _{A,t}\right) $ \\ 
Combined technology & $Z_{t}=A_{t}^{\frac{1}{1-\alpha }}MU_{t}^{\frac{\alpha 
}{1-\alpha }}$ \\ \hline\hline
\end{tabular}

\section{Model Parameters and Empirics\label{App Params}}

\begin{center}
\begin{tabular}{llcc}
\multicolumn{4}{c}{Table C1: $\underset{}{}$Model Parameters} \\ \hline\hline
Parameters & Description & FGR2015 & Estimated \\ 
&  & Values & Values \\ \hline
&  & (1) & (2) \\ \hline
$\beta $ & Time Preference & 0.99 & 0.9992 \\ 
$h$ & Habit Formation & 0.9 & 0.92747 \\ 
psi & labor supply coeff in utility & 8.0 &  \\ 
vartheta & Disutilty of Labor Scaling & 1.17 &  \\ 
$\delta $ & Depreciation Rate & 0.025 &  \\ 
$\alpha $ & Captial Share in Production & 0.21 & 0.14991 \\ 
$\kappa $ & Weight on Investment Adjustment Costs & 9.5 & 3.7946 \\ 
$\varepsilon $ & Elast. of Subst. btw. Differntiated Goods & 10 &  \\ 
eta & Elast. of Subst. btw Diff. Types of Labour & 10 &  \\ 
phi2 & Weight on Adj. Costs for Capital Utilization & 0.001 &  \\ 
$\chi _{w}$ & Wage Indexation & 0.6340 &  \\ 
$\chi $ & Price Indexation & 0.6186 & 0.00011223 \\ 
$\theta _{w}$ & Probability of not changing wages & 0.6869 &  \\ 
$\theta _{p}$ & Probability of not changing prices & 0.8139 & 0.8379 \\ 
\hline\hline
\multicolumn{4}{r}{continued on the next page}%
\end{tabular}

\begin{tabular}{llcc}
\multicolumn{4}{c}{Table C1: $\underset{}{}$Model Parameters -- continued}
\\ \hline\hline
Parameters & Description & FGR2015 & Estimated \\ 
&  & Values & Values \\ \hline
&  & (1) & (2) \\ \hline
\multicolumn{4}{l}{Policy Parameters} \\ \hline
$\gamma _{r}\left( S_{P=1}\right) $ & Interest rate smoothing & 0.7855 & 
0.80302 \\ 
$\gamma _{r}\left( S_{P=2}\right) $ & Interest rate smoothing & 0.7855 & 
0.87472 \\ 
$\gamma _{y}\left( S_{P=1}\right) $ & Reaction to output growth & 
exp(-1.4034) & 0.41649 \\ 
$\gamma _{y}\left( S_{P=2}\right) $ & Reaction to output growth & 
exp(-1.4034) & 0.32918 \\ 
$\pi _{\text{targ}}$ & Inlfation target & 1.0005 & 1.0057 \\ \hline
\multicolumn{4}{l}{Persistence of Shocks} \\ \hline
$\rho _{d}$ & Consumption preference & 0.1182 & 0.72554 \\ 
$\rho _{\varphi }$ & Labor supply & 0.9331 & 0.92016 \\ 
$\lambda _{a}$ & Neutral technology & 0.0028 &  \\ 
$\lambda _{\mu }$ & Investment-specific technology & 0.0034 &  \\ 
$\rho _{g}$ & Government spending shock & 0.75 & 0.0024123 \\ \hline
\multicolumn{4}{l}{Standard Deviation of Shocks} \\ \hline
$\sigma _{d}\left( S_{V=1}\right) $ & Consumption Preference Regime 1 & 
exp(-1.9834) & 0.031211 \\ 
$\sigma _{d}\left( S_{V=2}\right) ^{\ast }$ & Consumption Preference Regime 2
& 2*exp(-1.9834) & 0.16311 \\ 
$\sigma _{\varphi }\left( S_{V=1}\right) $ & Labor Supply Regime 1 & 
exp(-2.4983) & 0.29752 \\ 
$\sigma _{\varphi }\left( S_{V=2}\right) ^{\ast }$ & Labor Supply Regime 2 & 
2*exp(-2.4983) & 0.20313 \\ 
$\sigma _{\mu }\left( S_{V=1}\right) $ & Investment-specific technology
Regime 1 & exp(-6.0283) & 0.003998 \\ 
$\sigma _{\mu }\left( S_{V=2}\right) ^{\ast }$ & Investment-specific
technology Regime 2 & 2*exp(-6.0283) & 0.0071897 \\ 
$\sigma _{a}\left( S_{V=1}\right) $ & Neutral technology Regime 1 & 
exp(-3.9013) & 0.037752 \\ 
$\sigma _{a}\left( S_{V=2}\right) ^{\ast }$ & Neutral technology Regime 2 & 
2*exp(-3.9013) & 0.048111 \\ 
$\sigma _{g}$ & Government Spending shock & exp(-3.9013) & 0.0074137 \\ 
$\sigma _{\xi }$ & Monetary Policy & exp(-6.000) & 0.0020286 \\ \hline\hline
\multicolumn{4}{l}{Note: $^{\ast }$ denotes a parameter which is calibrated
by the authors.}%
\end{tabular}

\begin{tabular}{llcc}
\multicolumn{4}{c}{Table C2: Data Sources} \\ \hline\hline
Data series & Description & Units & FRED \\ 
&  &  & series \\ \hline
&  & (1) & (2) \\ \hline
DY\_DATA & Output Growth & \%pa & A939RX0Q048SBEA \\ 
DP\_DATA & Inflation rate & \%pa & GDPDEF \\ 
R\_DATA & Federal Funds Rate & \%pa & FEDFUNDS \\ 
DW\_DATA & Wage Inflation & \%pa & COMPRNFB \\ 
DMU\_DATA & Relat. price of invest. goods & \%pa & PIRIC \\ \hline\hline
\multicolumn{4}{l}{Note: FRED database https://fred.stlouisfed.org/}%
\end{tabular}
\end{center}

\begin{figure} [!h]
    \begin{center}
    \includegraphics[width = 0.95\textwidth]{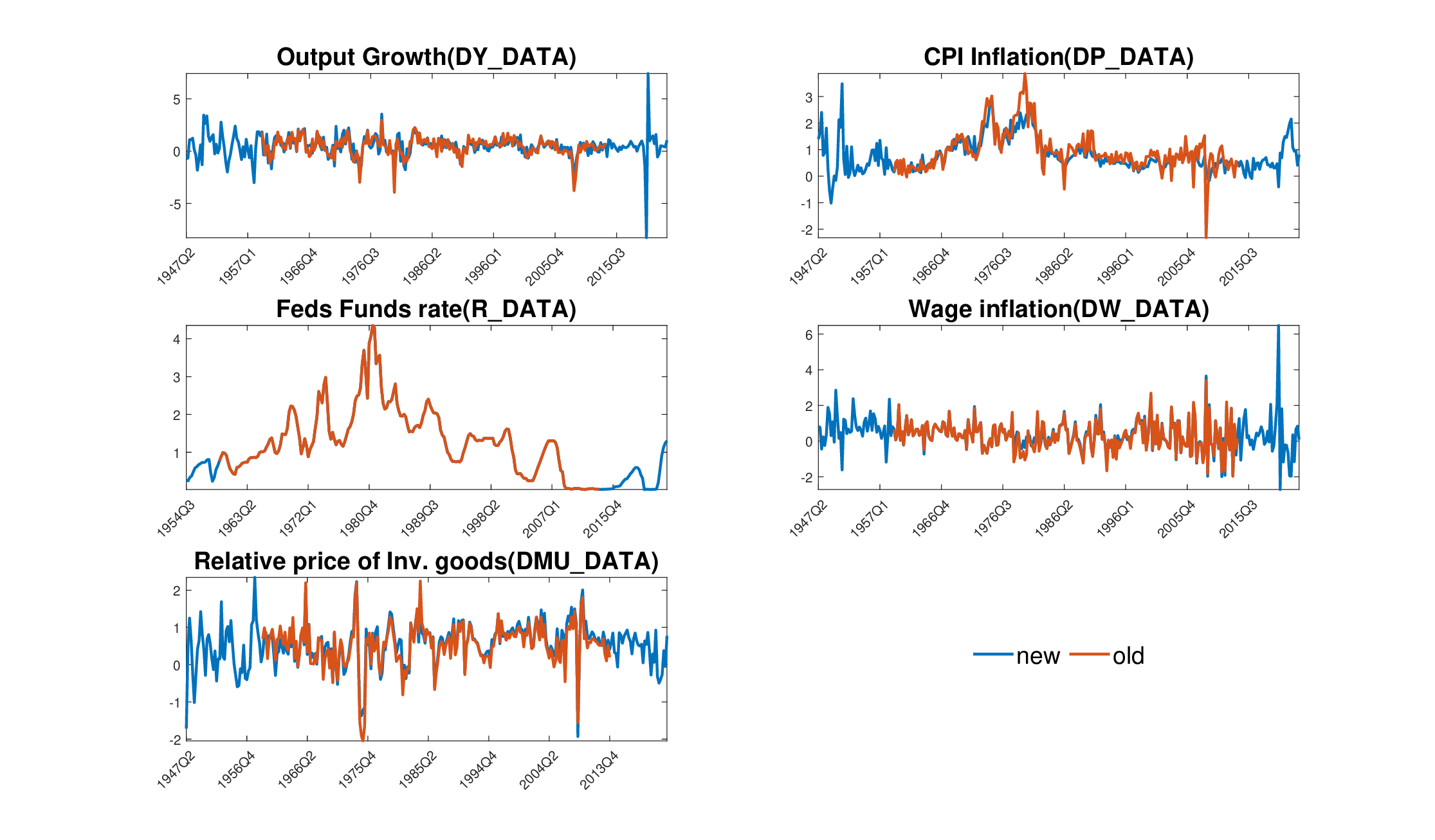}
    \caption{Updated Data.}
    \label{fig:FigureC1}
    \end{center}
\end{figure}

\section{Notes on Implementation\label{App RISE}}

All computations in this paper were coded in RISE$^{\textcopyright}$ (Maih,
2015). RISE toolkit accepts the model description as a text file containing
list of commands and mathematical expressions, converts it into a
state-space form, loads the data, and applies filters and smoothers
discussed in this paper.

In panel A of Table \ref{Tab Speed}, the speed results were obtained with
implementing IMM(1) and KN-GPB(2) as stand-alone procedure, with
optimisation for speed where possible. As discussed in the text, their
implementation in RISE allows handling various non-linearities and missing
observations.

In panel B of Table \ref{Tab Speed}, for comparison of speed within either
GPB or IMM\ families we employ single GPB(N) and IMM(N) filters that accept
an arbitrary order as input, rather than separate codes for different orders
of filtration. Consequently, the number of nested loops is not predetermined
but is managed throughout the computation.

\end{document}